\newif\ifconfver
\confverfalse      %declaring conference version false
\newif\ifcutshort      %this level shortens the equations
\cutshorttrue

\newif\ifcutshortlvltwo  %this level takes out some examples, figs., and sim.
\cutshortlvltwofalse

\ifconfver
\documentclass[10pt,twocolumn,twoside]{IEEEtran}
\else
\documentclass[10pt,twocolumn,twoside]{IEEEtran}
\fi

\usepackage{array}
\usepackage{cite}
\usepackage{calc}
\usepackage{caption}
\usepackage{graphicx}
\usepackage{psfrag}
\usepackage{stfloats}
\usepackage{url}
\usepackage{subfig}
\usepackage{longtable}
\usepackage{amsmath,amsfonts,amssymb,amsbsy,bm,paralist,theorem,ifthen,color}
\usepackage{mathtools}
\usepackage{algorithmic,algorithm}
\usepackage[dvips]{epsfig}
\usepackage[dvips]{graphicx}
\usepackage{caption}
\usepackage{multicol}
\usepackage{multirow}
\usepackage[font={small}]{caption}
\definecolor{orange}{RGB}{255,107,0}

\definecolor{green}{RGB}{0,180,80}

\newcommand{\Bc}{\mathcal{B}}

\newcommand{\Mc}{\mathcal{M}}

\newcommand{\Pc}{\mathcal{P}}
\newcommand{\Qc}{\mathcal{Q}}
\newcommand{\Rc}{\mathcal{R}}
\newcommand{\Sc}{\mathcal{S}}
\newcommand{\Tc}{\mathcal{T}}

\newcommand{\Xc}{\mathcal{X}}
\newcommand{\Yc}{\mathcal{Y}}
\newcommand{\Zc}{\mathcal{Z}}

\newcommand{\Cbb}{\mathbb{C}}
\newcommand{\Ebb}{\mathbb{E}}

\newcommand\Ab{\ensuremath{{\bf A}}}
\newcommand\Bb{\ensuremath{{\bf B}}}

\newcommand\Wb{\ensuremath{{\bf W}}}

\newcommand\Xb{\ensuremath{{\bf X}}}

\newcommand\pb{\ensuremath{{\bf p}}}

\newcommand\wb{\ensuremath{{\bf w}}}

\newcommand\taub{\ensuremath{{\bm \tau}}}

\newcommand\oneb{\ensuremath{{\bm 1}}}

\newcommand\alphab{\boldsymbol\alpha}

\newcommand{\CN}{\mathcal{CN}}

\DeclareMathOperator*   {\E}{\mathbb{E}}

\DeclareMathOperator*   {\argmax}{arg\,max}
\newtheorem{theorem}{Theorem}
\newtheorem{Definition}{Definition}

\newtheorem{proof}[theorem]{Proof}

\graphicspath{{fig/}}
\begin{document}
	
\bibliographystyle{IEEEtran}

\title{Max-Min Fairness User Scheduling and Power Allocation in  Full-Duplex OFDMA Systems}

\ifconfver \else {\linespread{1.1} \rm \fi
\author{\IEEEauthorblockN{Xiaozhou Zhang\IEEEauthorrefmark{1},
Tsung-Hui Chang\IEEEauthorrefmark{2},
Ya-Feng Liu\IEEEauthorrefmark{3},\\
Chao Shen\IEEEauthorrefmark{1}, and
Gang Zhu\IEEEauthorrefmark{1},}

\thanks{Part of this work was presented in 2017 IEEE ICASSP, New Orleans, USA \cite{Zhang2017}.
The work of T.-H. Chang was supported in part by the NSFC, China, under Grant 61571385 and Grant 61731018, and in part by the Shenzhen Fundamental Research Fund under Grant No. ZDSYS201707251409055 and No. KQTD2015033114415450. The work of X. Zhang and C. Shen are supported by the National Key R$\&$D Program of China (2016YFE0200900), NSFC, China (61871027, U1834210), the State Key Laboratory of Rail Traffic Control and Safety (RCS2019ZZ002), and the Major projects of Beijing Municipal Science and Technology Commission (Z181100003218010). The work of Y. Liu is supported in part by National Natural Science Foundation of China (NSFC) Grant No. 11671419, 11571221, 11688101, and 11631013, and in part by Beijing Natural Science Foundation Grant L172020.

X. Zhang, C. Shen and G. Zhu are with the State Key Laboratory of Rail Traffic Control and Safety, Beijing Jiaotong University, Beijing, China (email: xzzhang@bjtu.edu.cn; shenchao@bjtu.edu.cn; gzhu@bjtu.edu.cn).

T.-H. Chang is with the School of Science and Engineering, The Chinese University of Hong Kong, Shenzhen, China, and also with the Shenzhen Research Institute of Big Data, Shenzhen, China (email: tsunghui.chang@ieee.org)

Y.-F. Liu is with the State Key Laboratory of Scientific and Engineering Computing, Institute of Computational Mathematics and Scientific/Engineering Computing, Academy of Mathematics and Systems Science, Chinese Academy of Sciences, Beijing, China (e-mail: yafliu@lsec.cc.ac.cn).
}}
		
\maketitle
\begin{abstract}
In a full-duplex (FD) multi-user network, the system performance is not only limited by the self-interference but also by the co-channel interference due to the simultaneous uplink and downlink transmissions.
Joint design of the uplink/downlink transmission direction of users and the power allocation is crucial for achieving high system performance in the FD multi-user network.
In this paper, we investigate the joint uplink/downlink transmission direction assignment (TDA), user paring (UP) and power allocation problem for maximizing the system max-min fairness (MMF) rate in a FD multi-user orthogonal frequency division multiple access (OFDMA) system. The problem is formulated with a two-time-scale structure where the TDA and the UP variables are for optimizing a long-term MMF rate while the power allocation is for optimizing an instantaneous MMF rate during each channel coherence interval. We show that the studied joint MMF rate maximization problem is NP-hard in general. To obtain high-quality suboptimal solutions, we propose efficient methods based on simple relaxation and greedy rounding techniques. Simulation results are presented to show that the proposed algorithms are effective and achieve higher MMF rates than the existing heuristic methods.
\end{abstract}
%---------------------------------------------------------------------------
\IEEEpeerreviewmaketitle

%==============================================================================
\section{Introduction}

The full-duplex (FD) system has drawn considerable attention in recent years owing to its potential of doubling the system throughput by enabling the devices to transmit and receive information signals at the same time and over the same frequency.
However, to build a practical FD communication system, there still exist several fundamental challenges in circuit design and signal processing.
The major challenge lies in mitigating the so-called self-interference (SI) that is caused by simultaneous signal transmission and reception.
Fortunately, the breakthroughs in analog and digital interference cancellation schemes \cite{Katti_ACM13,Alexandropoulos2017}
have made effective mitigation of the SI possible and it has been shown that a FD system may outperform the half duplex (HD) system in practical scenarios \cite{Cheng2013}. For example, the FD techniques have been deployed in the relay networks \cite{Olivo2016,Xing2017} and the WiFi systems \cite{EDCOFDWS,Duarte2014}.

The FD technique has also been considered in the cellular networks \cite{Cirik2017,Cirik2015a,SPAWC17,Chang2018,CirikTC16,SunTWC16}, where the FD base stations (BSs) can communicate with both the downlink and the uplink user equipments (UEs) simultaneously.
In such FD networks, the downlink UE receives not only the desired information signal from the BS but also the interference signal from the uplink UE.
Therefore, this new form of co-channel interference (CCI) becomes another bottleneck for achieving high system performance in the FD cellular networks.
It is found that due to the SI and the CCI, the uplink and the downlink transmissions, i.e., the beamforming, have to be jointly designed, making the design problems much more challenging than the ones in the conventional HD cellular networks. In addition to the beamforming schemes, judiciously scheduling users, i.e., assigning the uplink/downlink transmission direction of UEs and pairing the downlink UEs with the uplink UEs, can greatly improve the FD system performance. This is because if the CCI between an uplink UE and a downlink UE is strong, then it is undesirable to group them together as a downlink/uplink pair. Following this idea, references \cite{Ahn2016a,Shahsavari2017}
studied joint beamforming/power control and uplink/downlink UE selection algorithms for maximizing the network throughput.

This paper considers a FD multi-user orthogonal frequency division multiple access (OFDMA) system, where different channels (i.e., resource blocks (RBs)) may be occupied by different UEs, and one RB can at most serve one uplink UE and one downlink UE simultaneously due to the FD capability. In the FD OFDMA system, the CCI occurs between the uplink and the downlink UEs in the same RB. Therefore, careful selection of a pair of uplink and downlink UEs for each RB can significantly enhance the throughput performance of the FD systems.
The goal of this paper is to answer the following fundamental question: How to properly schedule UEs and allocate power in the multi-user FD OFDMA system to maximize the system performance? Scheduling UEs here refers to determine which subset of UEs should be downlink/uplink UEs, and how the downlink/uplink UEs should be paired and assigned to appropriate RBs for data transmission.

\subsection{Related Works}
In the FD cellular networks, multiple-antenna techniques have been applied to the FD systems to overcome the CCI \cite{Cirik2017,Cirik2015a,SPAWC17,Chang2018,CirikTC16,SunTWC16}.
Specifically, references \cite{Cirik2017,Cirik2015a,SPAWC17} studied the beamforming designs for mitigating the SI and the CCI in the multi-user systems, in order to maximize the network sum rate, the proportional fairness rate, and the max-min fairness (MMF) rate, respectively. Besides, \cite{Chang2018,CirikTC16,SunTWC16}
considered designing the beamformers which satisfy the quality-of-service constraints in the multi-user system or the multiple-input multiple-output interference channel.

Recently, there is a considerable number of research works focusing on the FD systems with orthogonal channels, such as OFDMA or time division multiple access (TDMA) systems. In the OFDMA system, since one RB can serve at most one uplink UE and one downlink UE, the problem of assigning the uplink/downlink UE pairs to each RB has been widely studied. Aiming at maximizing the network sum rate of a FD OFDMA system, reference \cite{Han2016} and \cite{Nam2015} respectively used a matching algorithm and heuristic iterative RB assignment strategy for UE paring and RB allocation. Power allocation of the BS and the UEs was obtained by heuristic water-filling strategy.
Reference \cite{Silva2016} developed a Lagrangian dual based algorithm for joint UE pairing and power allocation with the target of maximizing the MMF rate of the UEs.
While the works in \cite{Silva2016,Han2016,Nam2015} assumed that the uplink/downlink transmission direction of UEs have been given in a prior, reference \cite{Shen2013} presented a heuristic algorithm for jointly determining the transmission direction assignment (TDA) of UEs, UE pairing (UP) and power allocation for maximizing the network sum rate.
In our previous work \cite{Zhang2017},  for MMF rate maximization, we considered the joint TDA and UP optimization with fixed and uniform power allocation, and presented a linear program based relaxation-and-rounding method.

\subsection{Contributions}
In this paper, like \cite{Silva2016}, we consider a multi-user FD OFDMA system, in which a FD BS communicates with a set of HD UEs aiming to maximize the MMF rate. However, different from \cite{Silva2016} but following a similar idea in \cite{Shen2013}, we assume that the TDA of the UEs are undetermined and optimize the uplink/downlink TDA, UP and power allocation jointly. In contrast to most of the existing works and our previous work \cite{Zhang2017} but sharing the same idea as \cite{Shahsavari2017}, we formulate the joint design problem as a two-time-scale MMF rate maximization problem. In particular, by considering the fact that the TDA and the UP solutions should not change as frequently as the fast fading of the wireless channels, in the formulated two-time-scale problem, the TDA and the UP variables are optimized to maximize a long-term MMF rate averaged over fast fading channels, while power allocation is performed for maximizing the MMF rate during each coherence interval. Such a two-time-scale formulation is not only more practical but also reduces significant signaling overhead as well as computational burden since the TDA and the UP solutions do not need to change whenever the small-scale channel fading changes.

Our technical contributions are twofold. Firstly, we conduct a computational complexity analysis for the formulated two-time-scale joint MMF rate maximization problem and prove that the joint MMF rate maximization problem is NP-hard. While \cite{Han2016} has studied the computational complexity for a joint user paring and sum rate maximization problem in the FD system, the results therein cannot imply the computational complexity of our joint MMF rate maximization problem. In fact, as shown in \cite{Sun2015}, simply adding a constraint could make an original NP-hard problem become polynomial time solvable.
In our previous work \cite{Zhang2017}, we have stated that the joint MMF rate maximization problem is NP-hard.
In this paper, we present the detailed proof by building a polynomial time transformation from the 3-dimensional matching problem \cite{Arora2009}, which is known to be NP-complete, to the joint MMF rate maximization problem. The proof of the NP-hardness gives interesting insights that the joint MMF rate maximization problem has the worst complexity in the full load case, i.e., when the number of the RBs is exactly half of the number of UEs.

Secondly, building upon the method in \cite{Zhang2017}, we further propose efficient algorithms for solving the two-time-scale joint MMF rate maximization problem.
Our algorithm is based on the alternating optimization (AO) method \cite{Bertsekas1999}, \cite{Tehrani2016}.
Specifically, we handle the joint MMF rate maximization problem by solving two subproblems iteratively: one subproblem is to
optimize the TDA and UP variables with fixed power allocation based on a continuous relaxation of the binary TDA and UP variables
followed by a heuristic rounding procedure, and the other subproblem is to optimize the power allocation  by existing methods based on successive convex approximation (SCA)\cite{Razaviyayn2013b,Shen2014}. While the relaxation-and-rounding method seems to work quite well in \cite{Nam2015} for the sum rate maximization problem, for the joint MMF rate maximization problem, it may be too loose to provide meaningful feasible solutions.
A meaningful feasible solution for our problem should satisfy the HD transmission constraint (i.e., the transmission directions of the HD UEs are consistent across all RBs) and should yield a strictly non-zero MMF rate. The latter condition requires that all UEs are properly paired and assigned to at least one of the RBs. Unfortunately, the conventional relaxation algorithms for the binary TDA and UP variables may not always be valid. To ensure the UEs to satisfy the HD transmission constraint, we propose a heuristic two-stage procedure based on a continuous binary relaxation scheme, which first determines the TDA of UEs and then performs UP optimization.
To guarantee a non-zero MMF rate (which is often violated in the full load case), we further equip the proposed two-stage procedure with an iterative greedy rounding method and the $\ell_q$-norm regularization technique (where $0 < q < 1$) \cite{LiuSIAM16} in the non-full load and the full load case, respectively. Extensive simulation results are presented to examine the performance of the proposed algorithms under different load cases.

The remainder of the paper is organized as follows. Section \ref{sec:system} presents the system model and the problem formulation. Section \ref{sec:NPhard} shows the NP-hardness of the MMF rate maximization problem. Section \ref{sec:algor} presents the algorithms for efficiently solving the considered problem. Simulation results are presented in Section \ref{sec:sim} and the conclusion is drawn in Section \ref{sec:con}.
%==============================================================================
\section{System Model and Problem Formulation}\label{sec:system}
\subsection{System Model}
\begin{figure}
\linespread{1}
  \centering
   \includegraphics[width=0.49\textwidth]{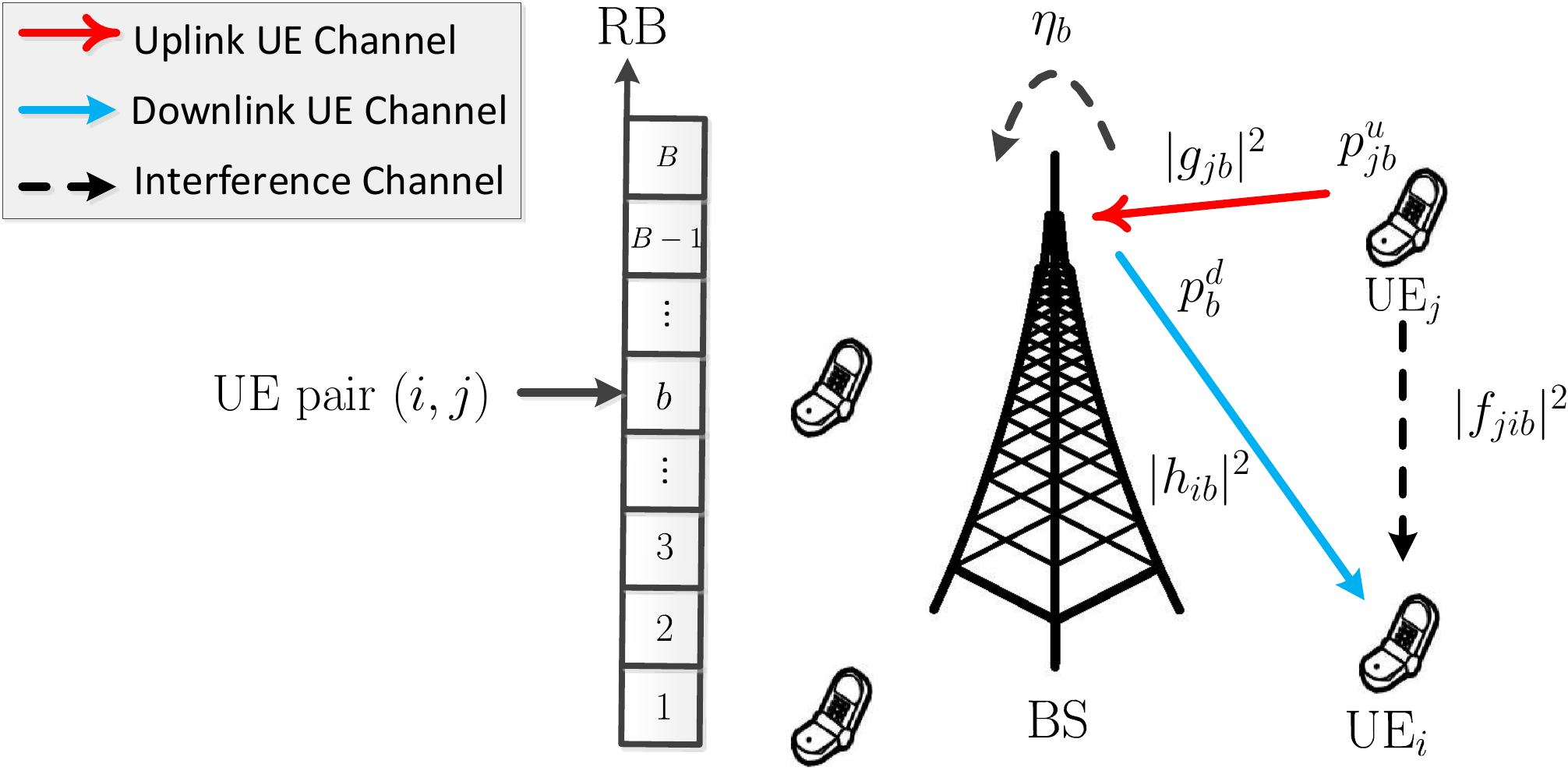}
  \caption{A single cellular OFDMA system with one FD BS and a set of HD UEs.
  	The FD BS has $B$ RBs and a downlink-uplink UE pair can be allocated to one RB.}
  \label{sysmod}
\end{figure}
As illustrated in Fig.~\ref{sysmod}, we consider an OFDMA system consisting of one (single-antenna) FD BS and $M$ (single-antenna) HD UEs. The system has $B$ frequency RBs. Owing to the FD ability, the BS can serve UEs for the uplink and downlink communications at the same time and frequency over each RB.
We assume that UEs always have data to transmit or receive; in other words, each UE can either work as an uplink UE or a downlink UE. We say that a downlink-uplink pair of UEs $(i,j)$ is allocated over RB $b$ if UE $i$ is a \emph{downlink} UE, UE $j$ is an \emph{uplink} UE and both of them are served by the BS over RB $b$. Following the OFDMA principle, we limit only one pair of UEs to be allocated to one RB. Throughout the paper, we assume that $M\le2B$; otherwise there exists at least one UE who can never be assigned to any RB.

Suppose that UE pair $(i,j)$ is allocated over RB $b$. We assume that the downlink and uplink transmissions of BS and UEs are synchronized\footnote{In particular, to obtain \eqref{equ:Rev_sgn}, UE $i$ needs to know the arrival time of the first multi-path signal component from the BS and the uplink UEs.
	Besides, the cyclic prefix (CP) length of the OFDM symbol is required to be longer than the time difference between the firstly arrived multi-path signal component and the lastly arrived multi-path signal component from the BS and the uplink UEs.
    Analogously, equation \eqref{equ:Tsm_sgn} holds as long as the BS knows the arrival time of the first multi-path signal component from the SI and the uplink UEs, and the CP length is longer than the time difference between the firstly arrived multi-path signal component and the lastly arrived multi-path signal component from the SI and the uplink UEs.}.
Then the signal received by the downlink UE $i$ over RB $b$ is given by
\begin{align}
  y_{ib}^{d}= \sqrt{ p_{b}^{d}}h_{ib}s_{ib}^{d}
             + \sqrt{p_{jb}^{u}}f_{jib}s_{jb}^{u} + n_{ib}^d, \label{equ:Rev_sgn}
\end{align}
where $p_{b}^d\ge0$ and $p_{jb}^u\ge0$ are the downlink transmission power of the BS and the uplink transmission power of UE $j$ over RB $b$, respectively;
$h_{ib}\in\Cbb$ is the downlink channel between the BS and UE $i$ over RB $b$, and
$f_{jib}\in\Cbb$ is the channel between UE $j$ and UE $i$ over RB $b$; $s_{ib}^d\in\Cbb$ and $s_{jb}^u\in\Cbb$ are the unit power signals (i.e., $\Ebb[{|s_{ib}^d|}^2] = 1$ and $\Ebb[{|s_{jb}^u|}^2] = 1$) transmitted from the BS to UE $i$ and from UE $j$ to the BS, respectively;
$n_{ib}^d\sim\CN(0,\sigma_i^2)$ represents the additive white Gaussian noise (AWGN) with zero mean, variance $\sigma_i^2$. As seen, the first term in the right-hand side (RHS) of \eqref{equ:Rev_sgn} is the desired signal of UE $i$, whereas the second term is the CCI from uplink UE $j$.

For the uplink, the signal received by the BS from UE $j$ over RB $b$ is given by
\begin{align}
  y_{jb}^{u} = \sqrt{p_{jb}^u}g_{jb}s_{jb}^u + \sqrt{ p_{b}^d}\hat{s}_{ib}^d + n_{b}^u,\label{equ:Tsm_sgn}
\end{align}
where $g_{jb}\in\Cbb$ is the uplink channel between UE $j$ and the BS over RB $b$, and $n_{b}^u\sim\CN(0,\sigma_0^2)$ is the AWGN at the BS. The first term in the RHS of \eqref{equ:Tsm_sgn} is the information signal from UE $j$, whereas the second term stands for the SI due to the FD BS.
Here we assume that the SI has been properly suppressed via some interference cancellation schemes \cite{Katti_ACM13,Alexandropoulos2017}. %\cite{Everett2014},
However, due to limited and non-ideal RF circuits \cite{Korpi2014,Masmoudi2016,Day2012}, the BS still suffers from the residual SI. Therefore, in \eqref{equ:Tsm_sgn},  we denote $\sqrt{ p_{b}^d} \hat s_{ib}^d$ as the residual SI term, where $\hat s_{ib}^d$ has an average residual SI channel gain
$\eta_b \triangleq  \Ebb[ |\hat s_{ib}^d|^2]$.

In this paper, we assume that the channel state information (CSI)  $h_{ib}$, $g_{jb}$ and $f_{jib}$ for all $i,j,b$ are known perfectly to the BS.
Specifically, when each uplink UE $j$ transmits the training signals, not only the BS can learn the uplink channel $g_{jb}$ but also the downlink UE $i$ can estimate $f_{jib}$ for all $b$ at the same time. Once each downlink UE $i$ obtains $f_{jib}$, it can send $f_{jib}$ and $h_{ib}$ for all $b$ back to the BS. More details about the CSI estimation schemes can be found in \cite{3GPP2018,Goyal2017}.
According to \eqref{equ:Rev_sgn} and \eqref{equ:Tsm_sgn}, when UE pair $(i,j)$ is allocated to RB $b$, the achievable downlink rate of UE $i$ and the uplink rate of UE $j$ are respectively given by
\begin{align}
R_{ijb}^{d}(\xi) &= \log_2\bigg(1+\frac{ p_{b}^d|h_{ib}|^2}{p_{jb}^u|f_{jib}|^2 + \sigma_i^2}\bigg),\label{equ:Thrputd}\\
R_{ijb}^{u}(\xi) &= \log_2\bigg(1+\frac{p_{jb}^u|g_{jb}|^2}{p_{b}^d\eta_b + \sigma_0^2 } \bigg),\label{equ:Thrputu}
\end{align}
where $\xi\triangleq \{h_{ib},f_{jib},g_{jb},\forall i,j,b\}$ denotes the quasi-static CSI of the system, which remains unchanged within the coherence time but can change from one block to another.
It is worthwhile noting from \eqref{equ:Thrputd} and \eqref{equ:Thrputu} that both the downlink and the uplink UEs suffer additional interference when compared to traditional HD OFDMA systems.
Specifically,
due to the CCI from uplink UE $j$ to downlink UE $i$, it is unwise to pair and allocate the two UEs to the same RB if the CCI channel $|f_{jib}|^2$ is stronger than the downlink channel $|h_{ib}|^2$.
This is because in that case, the downlink transmission power $p_{b}^d$ should be increased or the uplink transmission power $p_{jb}^u$ should be decreased for achieving a higher downlink rate in \eqref{equ:Thrputd}, which, however, would reduce the uplink transmission rate in \eqref{equ:Thrputu}.  Therefore, instead of pairing them, one should either allocate them to different RBs or change the two UEs to have the same transmission direction, i.e., both work as the uplink UEs or the downlink UEs, and pair them with another two UEs in different RBs.
This observation implies that the assignment of the uplink/downlink directions of the UEs and the UE paring
have significant impact on the system performance and have to be carefully designed. In the next subsection, we introduce the model for TDA and UP.
%==============================================================================
\subsection{Transmission Direction Assignment and UE Pairing}
Let $\Mc \triangleq \{1,2,...,M\}$ and $\Bc \triangleq \{1,2,...,B\}$ be the sets of UEs and RBs, respectively.
To determine whether a UE should work as an uplink UE or a downlink UE, i.e., TDA, we define a binary variable $\alpha_i\in\{0,1\}$. In particular, we let
\begin{align}
\alpha_i=
  \begin{cases}
  1,& \text{if UE $i$ is assigned to be a downlink UE,}\\
  0,& \text{otherwise}.
\end{cases}
\end{align}
To describe the UP and RB allocation, we define another binary variable $x_{ijb}\in\{0,1\}$ so that
\begin{align}
x_{ijb}=
  \begin{cases}
  1,& \text{if UE $i$ and UE $j$ are paired as $(i,j)$} \\
  & \text{and allocated to RB $b$,}\\
  0,& \text{otherwise}.
\end{cases}
\end{align}
The TDA and UP variables have to satisfy the following constraints:
\begin{itemize}
  \item \underline{OFDMA Constraint:} Under the OFDMA constraint, only one pair of UEs can be allocated to each RB $b$, i.e.,
    \begin{align}
        &\sum_{i\in\Mc}\sum_{j\in\Mc}x_{ijb} = 1,~b\in\Bc,\label{ctr:P_1pair}\\
        & x_{ijb}\in\{0,1\},~i,j\in\Mc,
        ~ b\in\Bc.\label{ctr:P_x_01}
    \end{align}
  \item \underline{HD Transmission Constraint:} Note that one UE could be assigned to more than one RB. However, since the UEs are HD, if the UE is assigned for uplink transmission (resp. downlink reception) in one RB, then it must also perform as an uplink (resp. downlink) UE over other RBs. To ensure this, we impose the following constraints
      \begin{align}
        &x_{ijb} \le \alpha_i,~i,j\in\Mc,~ b\in\Bc,\label{ctr:P_alpi}\\
        &x_{jib} \le 1-\alpha_i,~i,j\in\Mc,~ b\in\Bc,\label{ctr:P_alpj}\\
        &\alpha_{i} \in  \{0,1\},~ i\in\Mc.\label{ctr:P_alp_01}
    \end{align}
  \end{itemize}
    According to \eqref{ctr:P_alpi} to \eqref{ctr:P_alp_01}, when $\alpha_i=1$, the binary variable $x_{jib}= 0$ for all $i,j\in\Mc,b\in\Bc$, which implies that UE $i$ must consistently be a downlink UE for all allocated RBs. Similarly, when $\alpha_i=0$, UE $i$ can never be assigned as a downlink UE. Moreover, the variable $x_{iib}$ is always 0 for all $i\in\Mc$, $b\in\Bc$.
    In some scenarios, the UEs may not always have uplink data to transmit and downlink data to receive, or the TDA of the UEs has been predetermined.
    For example, if UE $i$ does not have downlink (resp. uplink) data, or UE $i$ has decided to be an uplink (resp. downlink) UE, one can deterministically set $\alpha_i=0$ (resp.  $\alpha_i=1$).
%==============================================================================
\subsection{Two-Time-Scale MMF Design}
With the definitions of the TDA and UP variables, we can express the achievable rate of each UE $i$ as
\begin{align}\label{equ:rate_i}
R_{i}(\xi)\triangleq\sum_{j\in\Mc}\sum_{b\in\Bc}
\left(x_{ijb}R_{ijb}^{d}(\xi)+x_{jib}R_{jib}^{u}(\xi)\right),i\in\Mc.
\end{align}
Note that, owing to the OFDMA constraint in \eqref{ctr:P_1pair} and the HD constraints in \eqref{ctr:P_alpi}-\eqref{ctr:P_alpj},
$R_{i}(\xi)$ in \eqref{equ:rate_i} is either the aggregate of the downlink rate or the aggregate of the uplink rate of UE $i$ over allocated RBs. In this paper, we aim to jointly optimize the TDA, UP and power allocation for maximizing an average system performance in a two-time-scale fashion. To explain this, let $U:\mathbb{R}^M \to \mathbb{R}$ be a rate utility function of the system (e.g., for the sum rate, $U(\{R_{i}(\xi)\}_{i=1}^M)=\sum_{i=1}^M R_{i}(\xi)$).
Then, we assume that the downlink and the uplink powers $( p_{b}^{d},~p_{jb}^{u})$ can be adapted at the same frequency as the CSI variation, that is, $(p_{b}^{d},~p_{jb}^{u})$ are designed to maximize the instantaneous rate utility function
\begin{align}
(\pb^{d}(\xi), \pb^{u}(\xi))=\arg\max_{\pb^{d}\in \Pc^d,~\pb^{u}\in \Pc^u}U(\{R_{i}(\xi)\}_{i=1}^M),
\end{align}
where $\pb^{d}$ and $\pb^{u}$ are the vectors respectively containing $ p_{b}^d$ and $p_{jb}^u$ for all $j\in\Mc,b\in\Bc$, and
\begin{align}
   &\Pc^d\triangleq \bigg\{p_{b}^{d}\geq 0,~\sum_{b\in\Bc}p_{b}^{d}\le P_{\rm BS},~b\in\Bc\bigg\},\label{ctr:P_p0max} \\
   &\Pc^u\triangleq \bigg\{ p_{jb}^{u}\geq 0,
   \sum_{b\in\Bc} p_{jb}^{u}\le P_{\rm UE},~j\in\Mc,~b\in\Bc
   \bigg\}, \label{ctr:P_pumax}
\end{align}
are respectively the feasible sets for the downlink and the uplink transmission powers; here $P_{\rm BS}$
and $P_{\rm UE}$ denote the maximum transmission powers of the BS and the UEs, respectively.

The TDA and UP are designed to maximize the expected rate utility; specifically, we are interested in the following problem
\begin{align}\label{eqn: P}
\max_{(\Xb,\alphab)\in \Qc} \E\nolimits_{\xi}
\bigg[\max_{\pb^{d}\in \Pc^d,~\pb^{u}\in \Pc^u}U(\{R_{i}(\xi)\}_{i=1}^M)\bigg],
\end{align} where $\Xb = \{\Xb_1,\ldots,\Xb_B\}$ is defined as a collection of UP variable matrix $\Xb_b$ with  $\Xb_b=\{x_{ijb}\}_{i,j\in\Mc}$ for all RB $b\in\Bc$; $\alphab$ is defined as a column vector of the TDA variables $\{\alpha_i\}_{i\in\Mc}$; $\Qc$ denotes the feasible set specified by the constraints in \eqref{ctr:P_1pair} to \eqref{ctr:P_alp_01}.
As seen, \eqref{eqn: P} is a two-time-scale design problem: the TDA and UP variables $(\Xb,\alphab)$ maximize the expected rate utility (long time scale) while the BS and the uplink UEs' transmission powers $(\pb^{d},\pb^{u})$ maximize the instantaneous rate utility (short time scale).
The two-time-scale design problem is meaningful for several reasons.
Firstly, it is reasonable that a UE usually has a bulk of data for either uplink or downlink transmission and TDA should not be frequently changed. Secondly, both TDA and UP largely depend on the large-scale channel conditions (e.g., the relative locations of the UEs and the BS) and therefore is not required to be adapted with the fast channel fading.
Thirdly, the two-time-scale design allows the reduction of the signaling overhead for re-assigning TDA and UP as well as reducing the system computational burden by avoiding from computing the TDA and the UP solutions whenever the channel fading changes.

The two-time-scale problem \eqref{eqn: P} is a stochastic optimization problem. It is in general difficult to solve as there is no closed-form expression for the objective function.
A common approach dealing with such difficulty is to employ the sample average approximation (SAA) method\cite{Birge2011,Verweij2003}, which approximates the expectation term by a sample average.
To illustrate this, denote $\xi^t\triangleq \{f_{jib}(t),~g_{ib}(t),~h_{ib}(t),~\forall~i,j,b\}$, $t\in\Tc=\{1,2,...,T\}$ to be a sequence of independently and identically distributed random CSI samples. Then we use the sample average of instantaneous rate function to approximate the expected rate, which leads to the following problem
\begin{align}\label{pro:sample_avg}
\max_{(\Xb,\alphab)\in \Qc}  \frac{1}{T}\sum_{t\in\Tc}
\bigg[\max_{\pb^{d}\in \Pc^d,~\pb^{u}\in \Pc^u}U(\{R_{i}(\xi^t)\}_{i=1}^M)\bigg].
\end{align}
It is worth mentioning that there exist theoretical results that characterize the approximation performance of the SAA method; interested readers may refer to \cite{Birge2011}.

While problem \eqref{pro:sample_avg} can accommodate any valid utility function $U$, in this paper we are particularly interested in the MMF rate, i.e.,
\begin{align}
  U(\{R_{i}(\xi^t)\}_{i=1}^M) =\min\{R_{1}(\xi^t)/\gamma_1,\ldots,R_{M}(\xi^t)/\gamma_M\},\nonumber
\end{align}
where $\gamma_1,\ldots,\gamma_M>0$ are some weights. In this case, problem \eqref{pro:sample_avg} becomes
\begin{align}\label{pro:sample_avg mmf}
\max_{(\Xb,\alphab)\in \Qc}  \frac{1}{T}\sum_{t\in\Tc}
\bigg[\max_{\pb^{d}\in \Pc^d,\,\pb^{u}\in \Pc^u}   \min_{i\in\Mc}  \left\{\frac{R_{i}(\xi^t)}{\gamma_i}\right\} \bigg].
\end{align}
Problem \eqref{pro:sample_avg mmf} maximizes the minimum (weighted) uplink/downlink rate of all UEs in the network and therefore ensures the fairness among the UEs.
From the perspective of the algorithm design, the MMF rate formulation \eqref{pro:sample_avg mmf} is more challenging than the formulation with other utilities since the MMF rate would be zero if there exits one UE who is not assigned to any one of the RBs. Indeed, from the complexity point of view, solving problem \eqref{pro:sample_avg mmf} is intrinsically difficult. We show in the next section that problem \eqref{pro:sample_avg mmf}  is in fact strongly NP-hard. In Section \ref{sec:algor}, we then propose some efficient approximation algorithms for solving problem \eqref{pro:sample_avg mmf}.

\section{Computational Complexity Analysis}\label{sec:NPhard}

\subsection{Brief Introduction to Complexity Theory}\label{subsectionintroduction}

In computational complexity theory{\cite{Complexitybook,combook2}}, a problem is said to be NP-hard if it is at least as hard as any problem in the class NP (problems that are solvable in Nondeterministic Polynomial time). NP-complete problems are the hardest problems in NP in the sense that if any NP-complete problem is solvable in polynomial time, then each problem in NP is solvable in polynomial time. A problem is strongly NP-hard (strongly NP-complete) if it is NP-hard (NP-complete) and cannot be solved by a pseudo-polynomial time algorithm. An algorithm that solves a problem is called a \emph{pseudo-polynomial} time algorithm if its time complexity function is bounded above by a polynomial function related to both of the {length} and the numerical values of the given data of the problem. This is in contrast to the polynomial time algorithm whose time complexity function depends only on the length of the given data of the problem. It is widely believed that, unless P$=$NP, there cannot exist a polynomial time algorithm to solve any NP-complete, NP-hard, or strongly NP-hard problem.

The standard way to prove an optimization problem is NP-hard is to establish the NP-hardness of its corresponding feasibility problem or decision problem. The latter is the problem to decide whether the global minimum (maximum) of the optimization problem is below (above) a given threshold or not. To show a decision problem $\mathcal{P}_2$ is NP-hard,
we usually follow three steps: 1) choose a suitable NP-complete decision problem $\mathcal{P}_1;$ 2) construct a {polynomial
time} transformation from any instance of $\mathcal{P}_1$ to an instance of $\mathcal{P}_2;$
3) prove under this transformation that any instance of problem
  $\mathcal{P}_1$ is true if and only if the constructed instance of problem $\mathcal{P}_2$ is true. See \cite{Complexitybook,combook2} for more details on complexity theory.

\subsection{Strong NP-Hardness of Problem \eqref{pro:sample_avg mmf}}\label{subsectionproblem}
In this subsection, we show that problem \eqref{pro:sample_avg mmf} is strongly NP-hard. The strong NP-hardness proof of problem \eqref{pro:sample_avg mmf} is based on a polynomial time reduction from the 3-dimensional matching problem \cite{Arora2009}, which is known to be NP-complete.
\begin{Definition} \textbf{(3-Dimensional Matching Problem with Size $K$)} \label{3-match}
	Given three sets $\Xc$, $\Yc$, $\Zc$ with $|\Xc|=|\Yc|=|\Zc|=K$ and a subset $\Sc\subseteq\Xc\times\Yc\times\Zc$. The 3-dimensional matching problem is to check whether there exists a subset $\Rc \subseteq\Xc\times\Yc\times\Zc$ satisfying
	\begin{enumerate}[(C1)]
		\item $\Rc\subseteq \Sc$;
		\item $|\Rc| = K$;
		\item For any two different triples $(i_x, j_y, \ell_z)\in \Rc$, $(i'_x,j'_y,\ell'_z)\in \Rc$, we have $i_x\neq i'_x$, $j_y\neq j'_y$, and $\ell_z\neq \ell'_z$.
	\end{enumerate}
   The subset $\Rc$ satisfying (C1), (C2), and (C3) is called a 3-dimensional match.
\end{Definition}

We are now ready to present the main result in this section.
\begin{theorem}\label{them:NP-hard}
	The MMF rate maximization problem \eqref{pro:sample_avg mmf} is strongly NP-hard.
\end{theorem}
\begin{proof} To show the theorem, it suffices to show that the following decision version of problem \eqref{pro:sample_avg mmf} with $T=1$ is strongly NP-hard:
\begin{equation}\label{pro:fea}
	\left\{\!\!\!\!\!\!\begin{array}{rl}
	& \displaystyle \sum_{j\in\Mc}\sum_{b\in\Bc}
\left(x_{ijb}R_{ijb}^{d}+x_{jib}R_{jib}^{u}\right)
	\ge \tau \gamma_i,~i\in\Mc,\\
	&(\Xb,\alphab)\in \Qc,~\pb^{d}\in \Pc^d,~\pb^{u}\in \Pc^u,
	\end{array}\right.
\end{equation}
where $\tau$ is a constant (which will be specified later). Since $T=1$, we have removed $\xi$ or $\xi^t$ in \eqref{pro:fea} for notational simplicity.

Given any instance of the 3-dimensional matching problem with size $K$, below we construct a corresponding instance of \eqref{pro:fea}. More specifically, given the three sets
$$\Xc=\{1_x,\ldots,K_x\},~\Yc=\{1_y,\ldots,K_y\},~\Zc=\{1_z,\ldots,K_z\},$$
and a subset
$$\Sc=\left\{(i_x,j_y,\ell_z)\mid i_x\in\Xc, j_y\in\Yc, \ell_z\in\Zc\right\}\subseteq \Xc\times\Yc\times\Zc,$$ let the set of UEs $\Mc$ be the union of $\Xc$ and $\Yc$ and let the set of RBs $\Bc$ be $\Zc$, i.e.,
\begin{equation}\label{setXYZ}
  \Mc=\Xc\cup\Yc~\text{and}~\Bc=\Zc.
\end{equation}
Hence,
\begin{equation}\label{BM}
B=K~\text{and}~M=2K.
\end{equation}
Without loss of generality, we let the set of the first $B$ elements in $\Mc$ be $\Xc$ and
the set of the last $B$ elements in $\Mc$ be $\Yc$. Next we set the channel coefficients as follows:
\begin{align}
&|h_{ib}|^2= 1, ~i\in\Mc,   ~b\in\Bc,\\
&|g_{jb}|^2= 1, ~j\in\Mc,   ~b\in\Bc,\\
&|f_{jib}|^2= |f_{ijb}|^2=
\begin{cases}
0,&\text{if $(i,j,b)\in \Sc$,}\\
1,&\text{otherwise},
\end{cases}\label{equ:defS}
\end{align}
set $\eta_b=1,b\in\Bc$, $\sigma_0=1$, $\sigma_i=1,i\in\Mc.$
Then, for each $(i,j,b)\in\Mc\times\Mc\times\Bc,$ the downlink rate $R_{ijb}^d$ for UE $i$ and the uplink rate $R_{ijb}^u$ for UE $j$ are respectively given by
\begin{align}\label{downlinkrate}
R_{ijb}^{d}&=\log_2\bigg(1+\frac{p_{b}^{d}}{p_{jb}^{u}|f_{jib}|^2 + 1}\bigg)\nonumber\\
&=\begin{cases}
\log_2\bigg(1+p_{b}^{d}\bigg),
&\text{if}~(i,j,b)\in \Sc,\\
\log_2\bigg(1+\frac{p_{b}^{d}}{p_{jb}^{u} + 1}\bigg), &\text{otherwise},
\end{cases}\\
R_{ijb}^{u}&=
\log_2\bigg(1+\frac{p_{jb}^{u}}{p_{b}^{d}+1}\bigg). \label{uplinkrate}
\end{align}
Moreover, set $\tau=1,$ $\gamma_i=1,~i\in\Mc$, $P_{\rm BS}=B$, and $P_{\rm UE}=2$. Hence, the constructed special instance of problem \eqref{pro:fea} becomes
\begin{equation}\label{pro:fea2}
	\left\{\!\!\!\!\!\!\!\begin{array}{rl}
	& \displaystyle \sum_{j\in\Mc}\sum_{b\in\Bc}
\left(x_{ijb}R_{ijb}^{d}+x_{jib}R_{jib}^{u}\right)
	\ge 1,i\in\Mc,\\
    & \displaystyle (\Xb,\alphab)\in \Qc,\\
    &\displaystyle \sum_{b\in\Bc} p_{jb}^{u}\le 2,~p_{jb}^{u}\geq 0, ~j\in\Mc,~b\in\Bc,	\\
    &\displaystyle \sum_{b\in\Bc}p_{b}^{d}\le B,~p_{b}^{d}\geq 0,~b\in\Bc,
	\end{array}\right.
\end{equation}
where $R_{ijb}^{d}$ and $R_{ijb}^{u}$ are given in \eqref{downlinkrate} and \eqref{uplinkrate}, respectively. We are going to show that the answer to the 3-dimensional matching
problem is yes if and only if the constructed problem \eqref{pro:fea2} is feasible.

We first show that if the answer to the 3-dimensional matching problem is yes, then the constructed problem \eqref{pro:fea2} is feasible. Suppose that there exists a 3-dimensional match $\Rc$ satisfying (C1) to (C3). Then we set
\begin{subequations}
\begin{align}
\tilde x_{ijb}&=
\begin{cases}
1,&\!\!\!\!\text{if}~(i,j,b)\in\Rc,\\
0,&\!\!\!\!\text{otherwise},
\end{cases}~(i,j,b)\in\Mc\times\Mc\times\Bc,\\
\tilde\alpha_i&=1,~i\in\Xc,~\tilde\alpha_j=0,~j\in\Yc,~\tilde p_{b}^{d} = 1,~b\in\Bc,\\
\tilde p_{jb}^{u}&=
\begin{cases}
2,&\!\!\!\!\text{if there exists an $i$ such that} (i,j,b)\in\Rc,\\
0,&\!\!\!\!\text{otherwise},
\end{cases}\\
&\text{for}~j\in\Mc,~b\in\Bc.\nonumber
\end{align}
\end{subequations}
It is simple to check that the above $(\tilde \Xb,\tilde \alphab,\tilde \pb^{d},\tilde \pb^{u})$ is a feasible solution to problem \eqref{pro:fea2}.

For the converse part, assuming that $(\tilde \Xb,\tilde \alphab,\tilde \pb^{d},\tilde \pb^{u})$ is a feasible solution to problem \eqref{pro:fea2}, we claim that the answer to the 3-dimensional matching problem is yes.
Define
\begin{align}\label{Rc}
\Rc=\{(i,j,b)\in\Mc \times \Mc \times \Bc~|~\tilde x_{ijb}=1 \}.
\end{align} Next, we prove that $\Rc$ is a 3-dimensional match satisfying (C1), (C2), and (C3) in Definition \ref{3-match}.

We first prove that $\Rc\subseteq\Sc$, i.e., (C1) is true. We prove this by contradiction.
Suppose that $\Rc\nsubseteq\Sc$. Then there must exist one triple $(i^\prime,j^{\prime},b^{\prime})\in \Mc \times \Mc\times \Bc$ such that $(i^{\prime},j^{\prime},b^{\prime})\in\Rc$ but $(i^{\prime},j^{\prime},b^{\prime})\notin\Sc$. By \eqref{equ:defS} and \eqref{Rc}, we have $|f_{j^{\prime}i^{\prime}b^{\prime}}|=1$ and $\tilde x_{i^{\prime}j^{\prime}b^{\prime}}=1$. Since $M=2B$ (c.f. \eqref{BM}) and all UEs' transmission rates are greater than or equal to $1$ (from our assumption that \eqref{pro:fea2} is feasible), each UE must occupy exactly one RB.
Consequently, for the UE $i^\prime$ and the UE $j^\prime$, we must have
\vspace{-1em}
$$R_{i^\prime j^\prime b^\prime}^{d}=\log_2\bigg(1+\frac{\tilde p_{b^\prime}^{d}}{\tilde p_{j^\prime b^\prime}^{u}+1} \bigg)\ge1$$
$$R_{i^\prime j^\prime b^\prime}^{u}=\log_2\bigg(1+\frac{\tilde p_{j^\prime b^\prime}^{u}}{\tilde p_{b^\prime}^{d}+1}\bigg)\ge1.$$
which further implies $\tilde p_{b^\prime}^{d}\ge \tilde p_{b^\prime}^{d}+2.$
This is a contradiction. Hence, $\Rc\subseteq\Sc$ and (C1) is true. This, together with the definition of $\Sc,$ immediately shows that $\Rc$ in \eqref{Rc} can be equivalently rewritten as
\begin{align}\label{Rc2}
\Rc=\{(i,j,b)\in\Xc \times \Yc \times \Bc~|~\tilde x_{ijb}=1 \}.
\end{align} Combining the above with the feasibility of $(\tilde \Xb,\tilde \alphab,\tilde \pb^{d},\tilde \pb^{u})$ further yields
\begin{align}
&\sum_{j\in\Yc}\sum_{b\in\Bc}\tilde x_{ijb}\log_2\bigg( 1+\frac{\tilde p_{b}^{d}}{ \tilde p_{jb}^{u}|f_{jib}|^2+1}\bigg)
\ge1,~i\in\Xc,\label{equ:fea_rate_X}\\
&\sum_{i\in\Xc}\sum_{b\in\Bc}\tilde x_{ijb}\log_2\bigg(1+\frac{\tilde p_{jb}^{u}}{\tilde p_{b}^{d}+1}\bigg)
\ge 1,~j\in\Yc,\label{equ:fea_rate_Y}\\
&\sum_{i\in\Xc}\sum_{j\in\Yc}\tilde x_{ijb} = 1,~b\in\Bc.\label{equ:fea_sum1}
\end{align}

We now show that $\Rc$ in \eqref{Rc2} satisfies (C2), i.e., $|\Rc|=K$. It follows from \eqref{equ:fea_sum1} that, for each $b\in\Bc$, there always exists a pair of $(i,j)\in\Xc\times \Yc$ such that $\tilde x_{ijb}=1$. As a result, there are totally $K=|\Bc|$ elements in the set $\Rc$ in \eqref{Rc2}, i.e., $|\Rc| = K$.

Finally, we show that $\Rc$ in \eqref{Rc2} satisfies (C3).
One can observe from  \eqref{equ:fea_rate_X} that, for each $i\in\Xc$, there exists at least one pair $(j,b)\in \Yc\times \Zc$ such that $\tilde x_{ijb}=1,$ which immediately implies
\begin{align}\label{Xset}
\Xc=\left\{ i \mid \text{there exists}~(j,b)\in\Yc\times\Bc~\text{such that}~(i,j,b)\in\Rc \right\}.
\end{align}
Similarly, from \eqref{equ:fea_rate_Y} and \eqref{equ:fea_sum1}, we respectively have
\begin{align}\label{Yset}
\Yc=\left\{ j \mid \text{there~exists}~(i,b)\in\Xc\times\Bc~\text{such~that}~(i,j,b)\in\Rc \right\},
\end{align}
Since $|\Xc|=|\Yc|=|\Bc|=|\Rc|=K$, it follows from \eqref{setXYZ}, \eqref{equ:fea_sum1}, \eqref{Xset} and \eqref{Yset} that, for any two different elements $(i_x,j_y,\ell_z)\in\Rc$ and $(i'_x,j'_y,\ell'_z)\in\Rc,$ we must have $i_x\neq i'_x$, $j_y\neq j'_y$, $\ell_z\neq \ell'_z$. Hence, $\Rc$ satisfies (C3).

It is simple to check that the above transformation from the 3-dimensional matching problem to the feasibility problem \eqref{pro:fea2} can be done in polynomial time. Since the 3-dimensional matching problem is strongly NP-complete, we conclude that
checking the feasibility of problem \eqref{pro:fea2} is strongly NP-hard.
Therefore, problem \eqref{pro:sample_avg mmf} is also strongly
NP-hard. \end{proof}

Two remarks on Theorem \ref{them:NP-hard} and its proof are in order. First, the proof of Theorem \ref{them:NP-hard} actually shows that problem \eqref{pro:sample_avg mmf} is strongly NP-hard even when $T=1$ (i.e., the single-time-scale formulation). By using the similar arguments as in the proof of Theorem \ref{them:NP-hard}, one can also show that problem \eqref{pro:sample_avg mmf} with only TDA and UP $(\Xb,\alphab)$ being optimization variables and the transmission powers $(\pb^d,\pb^u)$ being fixed is strongly NP-hard. Moreover, the proof of Theorem \ref{them:NP-hard} implies that the worst-case complexity of solving problem \eqref{pro:sample_avg mmf} happens when $M=2B$ (i.e., the full load case). Indeed, as is shown in Section \ref{sec:algor}, it is particularly challenging to design an efficient approximation algorithm for solving problem \eqref{pro:sample_avg mmf} in the full load case.

Second, the recent work \cite{Han2016} also studied the complexity analysis for a problem in the FD OFDMA system. However, the problem considered in \cite{Han2016} and our problem \eqref{pro:sample_avg mmf} are different. The key difference between the two problems lies in the requirement of UP and TDA. More specifically, the problem in \cite{Han2016} requires that each downlink (uplink) UE can only be paired with the same uplink (downlink) UE even over different RBs. For example, once UE $i$ is paired with UE $j$ over one RB, then UE $i$ cannot be paired with any other UEs over different RBs. This is in sharp contrast to our problem, where each downlink (uplink) UE can be flexibly paired with any uplink (downlink) UEs over different RBs. Moreover, the TDA (i.e., variable $\alphab$) is given and fixed in \cite{Han2016} but there is a freedom to design the TDA in our problem \eqref{pro:sample_avg mmf}. Another difference between the two problems is that the sum rate utility is adopted in \cite{Han2016} while the MMF rate utility is adopted in our problem. Therefore, the complexity results and techniques in \cite{Han2016} are not applicable to our problem \eqref{pro:sample_avg mmf}.
%--------------------------------------------------------------------
\vspace{-.9em}
\section{Proposed Algorithms}\label{sec:algor}
In this section, we present efficient algorithms for solving the two-time-scale problem \eqref{pro:sample_avg mmf}.
To the end, we assume that the statistical distribution of CSI $\xi$ is known, and a set of CSI $\xi^t,~ t=1,\ldots,T,$ are randomly generated for problem \eqref{pro:sample_avg mmf} in each time interval $T$.
As the problem involves two sets of variables $(\Xb,\alphab)$ and $\{(\pb^{d}(\xi^t), \pb^{u}(\xi^t)),~\forall~ \xi^t\}$, we jointly optimize the variables by adopting the AO method \cite{Bertsekas1999,Tehrani2016} to handle problem \eqref{pro:sample_avg mmf}.
Specifically, we handle the joint problem  \eqref{pro:sample_avg mmf} by {solving two subproblems iteratively}: one is to optimize the objective with respect to $(\Xb,\alphab)$ with fixed $\{(\pb^{d}(\xi^t), \pb^{u}(\xi^t)),~\forall~ \xi^t\}$, and the other one is to optimize the objective with respect to $\{(\pb^{d}(\xi^t), \pb^{u}(\xi^t)),~\forall~ \xi^t\}$ with fixed $(\Xb,\alphab)$.
In Section \ref{subsec:SR_SR} to Section \ref{subsec:IRM}, we focus on solving the first subproblem and propose a two-stage approach based on heuristic relaxation and iterative rounding techniques. In Section \ref{sec:WMMSE}, we consider optimizing the second subproblem.
\subsection{Performance of Simple Continuous Relaxation and Rounding}\label{subsec:SR_SR}
Considering \eqref{pro:sample_avg mmf} with power allocation variables $\{(\pb^{d}(\xi^t), \pb^{u}(\xi^t)),~\forall~ \xi^t\}$ fixed as below:
\begin{subequations}\label{pro:sample_avg mmf Xa}
\begin{align}
\max_{\substack{(\Xb,\alphab)\in \Qc,\\\taub}}~  &\frac{1}{T}\sum_{t\in\Tc} \tau_t\\
{\rm s.t.}~& R_{i}(\xi^t) \geq \tau_t \gamma_i,~i\in\Mc,~t\in\Tc,
\end{align}
\end{subequations}
where $\taub = \{\tau_t\}$ are the slack variables for the epigraph form. Problem \eqref{pro:sample_avg mmf Xa} has a linear objective function and linear constraints, but $(\Xb,\alphab)$ are binary variables.
A commonly adopted method for problem \eqref{pro:sample_avg mmf Xa} is to simply relax the binary variables to the continuous variables between zero and one, i.e., relax the constraints \eqref{ctr:P_x_01} and \eqref{ctr:P_alp_01} to
\begin{align}\label{relaxation}
0 \le x_{ijb} \le 1, ~i,j\in\Mc,b\in\Bc,~
0 \le \alpha_{i}\le 1,~i\in\Mc.
\end{align}
Since the continuous relaxation may be too loose to obtain a meaningful solution, it is useful to add some valid constraints (cuts) to the relaxed problem. For instance, \eqref{pro:sample_avg mmf Xa} would yield a non-zero MMF rate if and only if every UE is properly allocated to at least one RB and paired with some other UEs; that is, UP variable $\Xb$ must satisfy the pairing conservation constraint
\vspace{-.5em}
\begin{align}
&\sum_{j\in\Mc}\sum_{b\in\Bc}(x_{ijb} + x_{jib}) \ge 1,~i\in\Mc.\label{ctr:redunX}
\end{align}
Also, the OFDMA system requires that one RB can only be assigned to at most one uplink (or downlink) UE. Hence, the TDA variable $\alphab$ must satisfy
 \begin{align}
&\sum_{i\in\Mc}\alpha_i \le B ~\text{ and } \sum_{i\in\Mc} (1-\alpha_i) \le B.\label{ctr:redunA}
\end{align}

By incorporating \eqref{relaxation}, \eqref{ctr:redunX} and \eqref{ctr:redunA} into \eqref{pro:sample_avg mmf Xa}, we have
\begin{subequations}\label{pro:P1epi}
\begin{align}
\max_{(\Xb,\alphab,\taub)}  ~&\frac{1}{T}\sum_{t\in\Tc} \tau_t\\
{\rm s.t.}~& R_{i}(\xi^t) \geq \tau_t \gamma_i,~ i\in\Mc,~t\in\Tc, \\
& (\Xb,\alphab)\in \tilde \Qc,
\end{align}
\end{subequations}
where $\tilde \Qc$ contains the constraints in \eqref{ctr:P_1pair}, \eqref{ctr:P_alpi}, \eqref{ctr:P_alpj}, \eqref{relaxation}, \eqref{ctr:redunX} and \eqref{ctr:redunA}.
The above relaxed problem \eqref{pro:P1epi} is a linear program which can be efficiently solved by the off-the-shelf solvers. Once problem \eqref{pro:P1epi} is solved, the binary solutions may be obtained by simply rounding
$(\Xb,\alphab)$ back to the binary set. Specifically, due to \eqref{ctr:P_x_01}, for each RB $b$, one may set the largest value in $\Xb_b=\{x_{ijb}~|~i,~j\in\Mc\}$ to one and others to zero, while round the elements of $\alphab$ towards the nearest integer. Such method is referred to as simple relaxation (SR).

Unfortunately, the SR methods are likely to yield infeasible solutions. To illustrate this, we consider a simulation instance of \eqref{pro:P1epi} with $M=3$ and $B=2$, which has an optimal solution given by
\begin{align}
\alphab
&=
\begin{pmatrix}
0.65&0.3&0.35
\end{pmatrix};\label{xalpha}\\
%----------------------------
\Xb_1
&=
{\setlength{\arraycolsep}{2.5pt}
\begin{pmatrix}
0&0      &0.20\\
0&0      &0.30\\
0.35&0.15&0
\end{pmatrix};
%----------------------------
\Xb_2
=
\begin{pmatrix}
0   &0   &0.37\\
0   &0   &0.24\\
0.08&0.31&0
\end{pmatrix}.}\label{xalpha2}
\end{align}
\begin{itemize}
  \item  \underline{Violation of HD transmission constraint} \eqref{ctr:P_alpi} or \eqref{ctr:P_alpj}.

	Once $(\Xb,\alphab)$ in \eqref{xalpha},\eqref{xalpha2} is rounded, we obtain
    \begin{align}
	\alphab &=
	\begin{pmatrix}
	1& 0& 0
	\end{pmatrix};\nonumber\\
	%----------------------------
	\Xb_1 &=
	\begin{pmatrix}
	~0~&~0~&~0~\\
	~0~&~0~&~0~\\
	~1~&~0~&~0~
	\end{pmatrix};
	%----------------------------
    \Xb_2 =
	\begin{pmatrix}
	~0~&~0~&~1~\\
	~0~&~0~&~0~\\
	~0~&~0~&~0~
	\end{pmatrix}.\nonumber
	\end{align}
  As seen, the rounded solution violates \eqref{ctr:P_alpi} or \eqref{ctr:P_alpj} for $i=1$ and $i=3$.
  That is, UE 3 and UE 1 are assigned for downlink and uplink transmissions on RB 1 but the transmission directions are reversed on RB 2, which implies that the HD transmission constraint of UEs is violated. In fact, the violation probability of this constraint is quite high. As shown in Table \ref{tab:infea_NFL} (see column 2), for $M=8$ and $B\geq 16$, the probability that constraint \eqref{ctr:P_alpi} or \eqref{ctr:P_alpj} is violated is more than 85\% for the SR method.
  \end{itemize}
\begin{itemize}
  \item \underline{Violation of pairing conservation constraint} \eqref{ctr:redunX}.
	The above example shows that UE $2$ is never paired, which thus violates \eqref{ctr:redunX} and yields a zero MMF rate. As shown in Table \ref{tab:infea_NFL} (see column 5), there is a high probability that \eqref{ctr:redunX} is violated but it decreases when $B \geq M$.
\end{itemize}
\begin{table}[h]
	\centering
    \linespread{1}
	\caption{Comparison of the percentage of the solutions by the three algorithms (i.e., SR, 2S-SR and 2S-SRGR) which violating HD transmission constraints \eqref{ctr:P_alpi}, \eqref{ctr:P_alpj} or pairing conservation constraint \eqref{ctr:redunX}. The results are averaged over 200 channel realizations when $M=8$.}
	\label{tab:infea_NFL}	\begin{tabular}{|l|m{3.7em}m{3.6em}m{3.7em}|m{2.8em}m{2.8em}m{2.8em}|}
		\hline
		\multirow{2}{*}{$B$} & \multicolumn{3}{c|}{Violation of Constraint \eqref{ctr:P_alpi} or \eqref{ctr:P_alpj}} & \multicolumn{3}{c|}{Violation of Constraint \eqref{ctr:redunX}} \\
		\cline{2-7}
		&~SR & \hspace{-.7em}2S-SR & \hspace{-1.5em}2S-SRGR &~SR & \hspace{-.7em}2S-SR & \hspace{-1.5em}2S-SRGR \\
		\hline
		$4$  & 26\% & 0\%   & 0\%     & 98\% & 97\% & 45\%    \\
		$8$  & 68\% & 0\%   & 0\%     & 77\% & 88\% & 0\%     \\
		$16$ & 88\% & 0\%   & 0\%     & 49\% & 74\% & 0\%     \\
		$32$ & 96\% & 0\%   & 0\%     & 34\% & 56\% & 0\%     \\
		$64$ & 97\% & 0\%   & 0\%     & 26\% & 49\% & 0\%     \\
		\hline
	\end{tabular}
\vspace{-2em}
\end{table}

\subsection{Two-Stage Relaxation and Greedy Rounding}\label{subsec:SR_2stageSR}
We propose two heuristic strategies to overcome the infeasibility issue mentioned above.
Specifically, the reason for violating the HD transmission constraint \eqref{ctr:P_alpi} or \eqref{ctr:P_alpj} in the SR method is that the UP variables $\Xb$ and the TDA variables $\alphab$ are rounded independently. To fix this problem, we consider a \emph{two-stage} approach: in the first stage, we solve the relaxed problem \eqref{pro:P1epi}, and only round the TDA variables $\alphab$, denoted by $\alphab^\star$; in the second stage, we solve the relaxed problem \eqref{pro:P1epi} again but with the TDA variables being fixed to $\alphab^\star$ and with only the UP variables $\Xb$ as optimization variables, i.e., we solve the following problem in the second stage
\begin{subequations}\label{pro:P1epi fixed alpha}
	\begin{align}
	\max_{\Xb,\taub} ~&\frac{1}{T}\sum_{t\in\Tc} \tau_t\\
	{\rm s.t.}~& R_{i}(\xi^t) \geq \tau_t \gamma_i,~i\in\Mc,~t\in\Tc, \\
	& (\Xb,\alphab^\star)\in \tilde \Qc.
	\end{align}
\end{subequations}
The obtained solution from \eqref{pro:P1epi fixed alpha}, denoted by $\hat\Xb$, is then rounded by assigning the largest element in $\hat\Xb_b$ to one and others to zero for every $b\in \Bc$. Note that with the TDA variables $\alphab$ being fixed, the HD transmission constraint \eqref{ctr:P_alpi} or \eqref{ctr:P_alpj} must be satisfied for $\Xb$ in the second stage. We refer to this two-stage method as two-stage SR (2S-SR). As one can observe from Table \ref{tab:infea_NFL} (see column 3), the constrains \eqref{ctr:P_alpi} and \eqref{ctr:P_alpj} are always satisfied for the 2S-SR method.

However, the two-stage approach cannot guarantee that every UE is properly allocated to the RBs. Table \ref{tab:infea_NFL} (see column 6) shows that the solution returned by the 2S-SR method still violates constraint \eqref{ctr:redunX} with high probability. To fix this problem, we employ an iterative greedy rounding (GR) strategy in the second stage.
Specifically, unlike the 2S-SR method where all $\Xb_b,~b\in\Bc$ are rounded at once, in the GR strategy, we only round one $\Xb_b$ at a time. Let $\hat x_{i^\star j^\star b^\star}$ be the largest element in $\hat \Xb$ obtained by solving problem \eqref{pro:P1epi fixed alpha}. Then we set
$\hat x_{i^\star j^\star b^\star}=1$ and all the others $\hat x_{i j b^\star}=0$. The rounded $\hat \Xb_{b^\star}$ is then fixed in problem \eqref{pro:P1epi fixed alpha} and we solve the new problem (i.e., \eqref{pro:P1epi fixed alpha GR}) again with the rest UP variables. The above procedure is repeated until all RBs are visited. The details of the GR procedure is shown in line 3 to line 8 in Algorithm \ref{alg:GR}. We refer to the two-stage SR method with the iterative greedy rounding strategy as 2S-SRGR.
As shown in Table \ref{tab:infea_NFL} (column 7), the numerical results indicate that the 2S-SRGR method can always reduce the probability of violating constraint \eqref{ctr:redunX} to zero for the evaluated scenarios with $M< 2B$.
\setlength{\textfloatsep}{.5em}
\begin{algorithm}
\caption{Two-stage SR method with iterative greedy rounding strategy (2S-SRGR).}
\label{alg:GR}
\begin{algorithmic}[1]
    \STATE {\bf Stage One:} Solve the linear program \eqref{pro:P1epi} and obtain rounded $\alphab^\star$;
    \STATE {\bf Stage Two:} Set  $\tilde\Bc = \emptyset$;
        \REPEAT
            \STATE Solve the following linear program
            \begin{subequations}
            \label{pro:P1epi fixed alpha GR}
            	\begin{align}
            	\max_{\substack{\Xb_b,\forall b\in \Bc \backslash \tilde \Bc, \\\taub}}  ~&\frac{1}{T}\sum_{t\in\Tc} \tau_t\\
            	{\rm s.t.}~& R_{i}(\xi^t) \geq \tau_t \gamma_i,~i\in\Mc,~t\in\Tc, \\
            	& (\Xb,\alphab^\star)\in \tilde \Qc,
            	\end{align}
            \end{subequations}
          and obtain $\hat \Xb_b$, $b \in  \Bc \backslash \tilde \Bc$.
            \STATE Let $(i^\star,~j^\star,~b^\star)=
            \argmax\limits_{i,j\in\Mc,b\in\Bc\setminus\tilde\Bc}~~{\hat x_{ijb}}$;
            \STATE Obtain $\hat \Xb_{b^\star}^\star$ by setting $\hat x_{i^\star j^\star b^\star}=1$ and
            $\hat x_{ijb^\star}=0$ for all $(i,j)\neq (i^\star,j^\star)$;
            \STATE Update $\tilde\Bc  \leftarrow $ $b^\star \cup \tilde\Bc$;
        \UNTIL $\Bc = \tilde\Bc$;
    \STATE {\bf Output} $\hat \Xb_b^\star$, $b\in \Bc$ and $\alphab^\star$.
\end{algorithmic}
\end{algorithm}

\vspace{-1.6em}
\subsection{Tightening the Relaxation by Iterative Reweighted Minimization}\label{subsec:IRM}

As observed from Table \ref{tab:infea_NFL}, for $M=8$ and $B=4$, the probability that the solution returned by the 2S-SRGR method violates constraint \eqref{ctr:redunX} is still 45\%.
This is because for the full load case (i.e., when $M=2B$), each UE has to exactly occupy one RB; otherwise the MMF rate is zero. The complexity analysis in Section \ref{sec:NPhard} also suggests that the full load case is the most difficult case to handle. To overcome this issue, we propose to employ the $\ell_q$-norm regularization and the iterative reweighted minimization (IRM) method in \cite{LiuSIAM16,Liu2015} to tighten the relaxation in \eqref{relaxation}.

Specifically, according to \cite{LiuSIAM16,Liu2015}, the OFDMA constraints \eqref{ctr:P_1pair} and \eqref{ctr:P_x_01} can be well approximated by the simple relaxation \eqref{relaxation} plus a proper
$\ell_q-$norm regularization on $\Xb$. For the TDA variable $\alphab$, as $\alpha_i+(1-\alpha_i)=1,~\text{for all}~i\in\Mc$ always hold true, one can also approximate the binary constraint \eqref{ctr:P_alp_01} by \eqref{relaxation} and an
$\ell_q-$norm regularization of $\alphab$. Therefore, for the full load case, we propose to replace problem \eqref{pro:P1epi} in the first stage of Algorithm \ref{alg:GR} by the following regularized problem
\begin{subequations}\label{pro:P1epi lq}
	\begin{align}
	\max_{(\Xb,\alphab,\taub)}  ~&\frac{1}{T}\sum_{t\in\Tc} \tau_t - \rho_1 \|\Xb+\epsilon_1\|_q^q \nonumber\\
&\qquad\quad~ -\rho_2 (\|\alphab+\epsilon_2\|_q^q+\|\oneb-\alphab+\epsilon_2\|_q^q)\\
	{\rm s.t.}~& R_{i}(\xi^t) \geq \tau_t \gamma_i,~i\in\Mc,~t\in\Tc, \\
	& (\Xb,\alphab)\in \tilde \Qc,
	\end{align}
\end{subequations}
where
\begin{align}
  \|\Xb+\epsilon_1\|_q^q &\triangleq \sum_{i\in\Mc}\sum_{ j\in\Mc}\sum_{b\in\Bc}(x_{ijb}+\epsilon_1)^q, \nonumber\\
  \|\alphab+\epsilon_2\|_q^q+&\|\oneb-\alphab+\epsilon_2\|_q^q \nonumber\\ & \triangleq \sum_{i\in\Mc}\big((\alpha_i+\epsilon_2)^q + (1-\alpha_i+\epsilon_2)^q\big),\nonumber
\end{align}
$\rho_1\ge 0$, $\rho_2\ge 0$, $q\in(0,1)$, $\epsilon_1>0$ and $\epsilon_2>0$ are some parameters, and $\oneb$ is the all-one vector. It has been shown in \cite{Liu2015,LiuSIAM16} that for sufficiently large $\rho_1$ and $\rho_2$, \eqref{pro:P1epi lq} can asymptotically have the same optimal solution as \eqref{pro:P1epi}. Analogously, in the second stage of Algorithm \ref{alg:GR}, problem \eqref{pro:P1epi fixed alpha GR} is replaced by the following regularized problem
\begin{subequations}\label{pro:P1epi fixed alpha GR lq}
	\begin{align}
	\max_{\substack{\Xb_b,\forall b\in \Bc \backslash \tilde \Bc, \\\taub}}  ~&\frac{1}{T}\sum_{t\in\Tc} \tau_t-\rho_1 \|\Xb+\epsilon_1\|_q^q\\
	{\rm s.t.}~& R_{i}(\xi^t) \geq \tau_t \gamma_i,~i\in\Mc,~t\in\Tc, \\
	& (\Xb,\alphab^\star)\in \tilde \Qc,
	\end{align}
\end{subequations}
where $\alphab^\star$ is given by the rounded solution of \eqref{pro:P1epi lq}.

\setlength{\textfloatsep}{.5em}
\begin{algorithm}
\caption{IRM Algorithm for Solving Problem \eqref{pro:P1epi lq}}
\label{alg:FP_1sts_IRM}
\begin{algorithmic}[1]
    \STATE Given $q\in(0,1)$, $\rho_1\ge0$, $\rho_2\ge0$, $\epsilon_1\in(0,1)$,
    $\epsilon_2\in(0,1)$, $\sigma_1\in(0,1)$, $\sigma_2\in(0,1)$ and $\kappa>1$;
    \STATE Given a feasible solution $\alphab^{(0)}$, $\Xb^{(0)}$ obtained by solving \eqref{pro:P1epi};
    \STATE Set
    $\Wb^{(1)}=\left(\Xb^{(0)}+\epsilon_1\right)^{\circ(q-1)}$, \\ \quad~~$\wb^{a(1)}=\left(\alphab^{(0)}+\epsilon_2\right)^{\circ(q-1)}$,\\
    \quad~~$\wb^{b(1)}=\left(\oneb-\alphab^{(0)}+\epsilon_2\right)^{\circ(q-1)}$;
    \REPEAT
      \STATE Set $r = 1$;
      \REPEAT
        \STATE Obtain $\taub^{(r)}$, $\alphab^{(r)}$ and $\Xb^{(r)}$  by solving \eqref{pro:IRM1};
        \STATE Update
               $\Wb^{(r+1)}=\left(\Xb^{(r)}+\epsilon_1\right)^{\circ(q-1)}$,\\
               \qquad~~~
               $\wb^{a(r+1)}=\left(\alphab^{(r)}+\epsilon_2\right)^{\circ(q-1)}$,\\
               \qquad~~~
               $\wb^{b(r+1)}=\left(\oneb-\alphab^{(r)}+\epsilon_2\right)^{\circ(q-1)}$;
        \STATE $r \leftarrow r+1$;
      \UNTIL
             $\left\|\Xb^{(r)}-\Xb^{(r-1)}\right\|_1\le \sigma_1$ and \\ \qquad~$\left\|\alphab^{(r)}-\alphab^{(r-1)}\right\|_1\le \sigma_2$;
      \STATE $\rho_1\leftarrow \kappa\rho_1$,
             $\rho_2\leftarrow \kappa\rho_2$,
             $\epsilon_1\leftarrow\epsilon_1/\kappa$,
             $\epsilon_2\leftarrow\epsilon_2/\kappa$;
    \UNTIL $\max_{i\in\Mc}\min\{\alpha_i,1-\alpha_i\} \le \sigma_2$.
    \end{algorithmic}
\end{algorithm}
The $\ell_q$-norm regularized problem \eqref{pro:P1epi lq} and \eqref{pro:P1epi fixed alpha GR lq} are non-convex and difficult to solve in general. We apply the IRM algorithm in \cite{Candes2008} to solve the two problems. The IRM algorithm iteratively approximates the non-convex $\ell_q$-norm term by its first-order approximation. In particular, by taking \eqref{pro:P1epi lq} as the example, we solve the following problem at the $r$-th iteration of the IRM algorithm:
\begin{subequations}\label{pro:IRM1}
	\begin{align}
	\max_{(\Xb,\alphab,\taub)}
	~&\frac{1}{T}\sum_{t\in\Tc}\tau_t- \rho_1 q~ \left\|\Wb^{(r)}\circ\Xb\right\|_{\ell_1}\nonumber\\
	&-\rho_2 q\left\|(\wb^{a(r)}\circ\alphab + \wb^{b(r)}\circ(\oneb-\alphab))\right\|_1\\
	{\rm s.t.}~& R_{i}(\xi^t) \geq \tau_t \gamma_i,~i\in\Mc,~t\in\Tc, \\
& (\Xb,\alphab)\in \tilde\Qc,
	\end{align}
\end{subequations}
where $\Ab\circ\Bb$ is the Hadamard product of two tensors,
$\left\|\Ab\right\|_{\ell_1}$ is denoted as the summation of all elements of $\Ab$,
$\Ab^{\circ b}$ is the Hadamard power, $\Wb^{(r)}\triangleq\left(\Xb^{(r-1)}+ \epsilon_1\right)^{\circ(q-1)}$,
$\wb^{a(r)}\triangleq \left(\alphab^{(r-1)}+\epsilon_2\right)^{\circ(q-1)}$,
$\wb^{b(r)}\triangleq \left(\oneb-\alphab^{(r-1)}+\epsilon_2\right)^{\circ(q-1)}$,
$\Xb^{(r-1)}$ and
$\alphab^{(r-1)}$
respectively denote the value of the variables $\Xb$ and $\alphab$ at the $(r-1)$-th iteration.
Note that problem \eqref{pro:IRM1} is a linear program and thus can be efficiently solved. To be self-contained, we list the details of the IRM algorithm for solving problem \eqref{pro:P1epi lq} in Algorithm \ref{alg:FP_1sts_IRM}. The two-stage algorithm using $\ell_q$-norm regularization, and the IRM algorithm is referred to as 2S-IRMGR and is summarized in Algorithm \ref{alg:2S-IRMGR}. As will be shown in Section \ref{sec:sim}, the 2S-IRMGR method can effectively improve the MMF rate performance for the full load case.

\setlength{\textfloatsep}{.5em}
\vspace{-.3em}
\begin{algorithm}
\caption{ Two-Stage IRM method with iterative greedy rounding strategy (2S-IRMGR) }
\label{alg:2S-IRMGR}
\begin{algorithmic}[1]
    \STATE {\bf Stage One:} Solve the $\ell_q$-norm problem \eqref{pro:P1epi lq} by the IRM algorithm (e.g., Algorithm \ref{alg:FP_1sts_IRM}) and obtain rounded $\alphab^\star$;
\STATE {\bf Stage Two:} {\bf Set}  $\tilde\Bc = \emptyset$
\REPEAT
\STATE Solve the $\ell_q$-norm problem \eqref{pro:P1epi fixed alpha GR lq} by the IRM algorithm
and obtain $\hat \Xb_b$, $b \in  \Bc \backslash \tilde \Bc$.
\STATE Let $(i^\star,~j^\star,~b^\star)=
\argmax\limits_{i,j\in\Mc,b\in\Bc\setminus\tilde\Bc}~~{\hat x_{ijb}}$;
\STATE Obtain $\hat \Xb_{b^\star}^\star$ by setting $\hat x_{i^\star j^\star b^\star}=1$ and
$\hat x_{ijb^\star}=0$ for all $(i,j)\neq (i^\star,j^\star)$;
\STATE Update $\tilde\Bc  \leftarrow $ $b^\star \cup \tilde\Bc$;
\UNTIL $\Bc = \tilde\Bc$;
\STATE {\bf Output} $\hat \Xb_b^\star$, $b\in \Bc$ and $\alphab^\star$.
\end{algorithmic}
\end{algorithm}

\vspace{-1.5em}
\subsection{Optimization of Power Allocation}\label{sec:WMMSE}
The previous three subsections have focused on the optimization of the TDA and UP variables $(\Xb,\alphab)$ in problem \eqref{pro:sample_avg mmf} and assumed that the power allocation variables $\{(\pb^{d}(\xi^t), \pb^{u}(\xi^t)),~\forall~ \xi^t\}$ are fixed.
In this subsection, we consider the optimization of $\{(\pb^{d}(\xi^t), \pb^{u}(\xi^t)),~\forall~ \xi^t\}$ by assuming that $(\Xb,\alphab)$ are given. With $(\Xb,\alphab)$ being given, denoted by $(\Xb^\star,\alphab^\star)$, problem \eqref{pro:sample_avg mmf}
reduces to $T$ independent problems. Specifically, for each $t=1,\ldots,T$, we solve
\begin{subequations}\label{pro:P1epi fixed Xalpha}
	\begin{align}
	\max_{\substack{\pb^{d}\in \Pc^d,\\\pb^{u}\in \Pc^u}} ~& \tau_t\\
	{\rm s.t.} & \sum_{j\in\Mc}\sum_{b\in\Bc}
	x_{ijb}^\star R_{ijb}^{d}(\xi^t) \geq \tau_t \gamma_i, ~i\in\Mc \cap \{i|\alpha_i^\star=1\},\\
	&\sum_{j\in\Mc}\sum_{b\in\Bc}
	 x_{jib}^\star R_{jib}^{u}(\xi^t)\geq \tau_t \gamma_i, ~i\in\Mc\cap \{i|\alpha_i^\star=0\},
	\end{align}
\end{subequations}
where $R_{ijb}^{d}(\xi^t)$ and $R_{jib}^{u}(\xi^t)$ are given in \eqref{equ:Thrputd} and \eqref{equ:Thrputu}, respectively.
Problem \eqref{pro:P1epi fixed Xalpha} is a non-convex optimization problem. However, it can be efficiently handled by the existing SCA methods such as \cite{Razaviyayn2013b} and the weighted MMSE (WMMSE) based approaches \cite{Shi2011}.

According to AO, iterative updates of TDA and UP variables $(\Xb,\alphab)$ and power allocation variables $\{(\pb^{d}(\xi^t), \pb^{u}(\xi^t)),~\forall~ \xi^t\}$
by solving \eqref{pro:sample_avg mmf Xa} and \eqref{pro:P1epi fixed Xalpha} respectively can further improve the performance. However, it is found in the numerical experiments that the improvement is usually not significant and one round of updates of $(\Xb,\alphab)$ and $\{(\pb^{d}(\xi^t), \pb^{u}(\xi^t)),~\forall~ \xi^t\}$
is sufficient to obtain considerably good MMF rate performance, as will be verified in Section \ref{sec:sim}.

Once problem \eqref{pro:sample_avg mmf} is solved, we obtain the TDA and UP solution $(\Xb,\alphab)$ and employ it for the time interval $T$ regardless of the fast fadings. However, for each fading realization $\xi^{t'}, t'=1,\ldots,T,$ that occur in this time interval $T$, one can update and perform power control by solving problem \eqref{pro:P1epi fixed Xalpha} with respect to each $\xi^{t'}$. So the TDA and UP are performed and fixed in the large-time-scale $T$ while the power control is adapted according to the fast fadings. The performance of the two-time-scale formulation \eqref{pro:sample_avg mmf} and proposed algorithms are examined in the next section.

%==============================================================================
\section{Simulation Results}\label{sec:sim}
In this section, we evaluate the performance of the multi-user FD OFDMA system as described in Section \ref{sec:system}. Simulation parameters are set based on the 3GPP standard \cite{3GPP}.
We assume that the UEs are randomly and uniformly located within a circle centered at the BS with a radius $100$ m.
The maximum transmission power of all UEs $P_{\rm UE}$ is 23 dBm, and the total transmission power of the BS $P_{\rm BS}$ is $30$ dBm.
As for the channel model, the path loss component is given by $140.7+36.7\log_{10}(d)$ (dB), where $d$ (km) is the distance between the BS and the UE.
The small scale fading of all links are independently and identically generated following the complex Gaussian distribution with zero mean and unit variance.
At the BS, the residual SI channel gain $\eta_b$ is set to be $-110$ dB if not specified, and the noise power at the BS and UEs is set as $\sigma_0^2 = \sigma_i^2 = -90$ dBm for all $i\in\Mc$.
The number of channel samples in \eqref{pro:sample_avg mmf} is $T=100$.
All the results shown in this section were obtained by running 200 simulation trials. The parameters used in Algorithm \ref{alg:FP_1sts_IRM} and Algorithm \ref{alg:2S-IRMGR} are set as follows: $q=0.5$, $\rho_1 = \rho_2 = 1$, $\epsilon_1 =\epsilon_2 = 0.1$, $\sigma_1 = 10^{-3}$, $\sigma_2=0.1$, $\kappa=1.5$.

In addition to the four aforementioned methods, i.e., SR, 2S-SR, 2S-SRGR and 2S-IRMGR, for TDA and UP optimization, we also implement a heuristic algorithm following a similar idea as in \cite{Shen2013}. Specifically, for each UE $i$, we respectively compute the downlink rate $R_{ijb}^{d}$ and the uplink rate $R_{jib}^{u}$  (see \eqref{equ:Thrputd} and \eqref{equ:Thrputu}) for all $j\in \Mc$ and $b\in \Bc$, and obtain the average rates
$$\bar R_{i}^{d} = \sum_{j\in\Mc}\sum_{b\in\Bc}\frac{R_{ijb}^{d}}{(M-1)B}~\text{and}~
\bar R_{i}^{u} = \sum_{j\in\Mc}\sum_{b\in\Bc}\frac{R_{jib}^{u}}{(M-1)B}.$$
Then we assign UE $i$ to be a downlink UE (i.e., assign $\alpha_i=1$) if $\bar R_{i}^{d}\geq \bar R_{i}^{u}$ and an uplink UE (i.e., $\alpha_i=0$) otherwise. If the number of downlink UEs is larger than $B$, then we keep the TDA for the first $B$ downlink UEs that have larger values of $\bar R_{i}^{d}$ and change the rest to be uplink UEs. The same rule applies if the number of uplink UEs is larger than $B$. Given the two sets of uplink UEs and downlink UEs, we list all possible pairs of uplink and downlink UEs. From RB 1 to RB $B$, we select and allocate one pair of UEs to each of the RBs sequentially.
The selection criterion is based on the MMF rate among all the UEs that have not been allocated to at least one RB.
Specifically, for each RB, we choose the UE pair that contains an UE which has never been assigned to one RB but can increase the MMF rate most.
If all UEs have been assigned, then we simply choose the UE pair that can increase the MMF rate most. We refer to this method as ``Heuristic'' in the figures.

Furthermore, we also develop three algorithms modified from the UP algorithms A1-A3 in \cite{Alexand2016}.
Specifically, we add an extra criterion that the unpaired UEs always have higher priority to be assigned to the RB than those UEs who have already occupied another RB.
By this way, the algorithms can guarantee every UE a non-zero MMF rate. Before applying the UP algorithms, the TDA of the UEs are determined by the above Heuristic method.
We refer the new algorithms as ``Modified Algorithm A1 from \cite{Alexand2016}'',  ``Modified Algorithm A2 from \cite{Alexand2016}'' and  ``Modified Algorithm A3 from \cite{Alexand2016}''.

\subsection{Case of $M<2B$}
We first consider the case of $M<2B$, and examine the performance of TDA and UP optimization algorithms (SR, 2S-SR, 2S-SRGR, Heuristic and the modified algorithms A1-A3 from \cite{Alexand2016}).
The power allocation is fixed and set to uniform values as $p_{b}^d(t) = \frac{P_{\rm BS}}{B}$ and $p_{jb}^u(t) = \frac{P_{\rm UE}}{B}$ for all $i,j\in\Mc, b\in \Bc$ and $t\in\Tc$.
\begin{figure}
	\center
    \linespread{1}
	\subfloat[$M=8,~B=16$]
	{\label{fig2a}\includegraphics[width=0.39\textwidth]{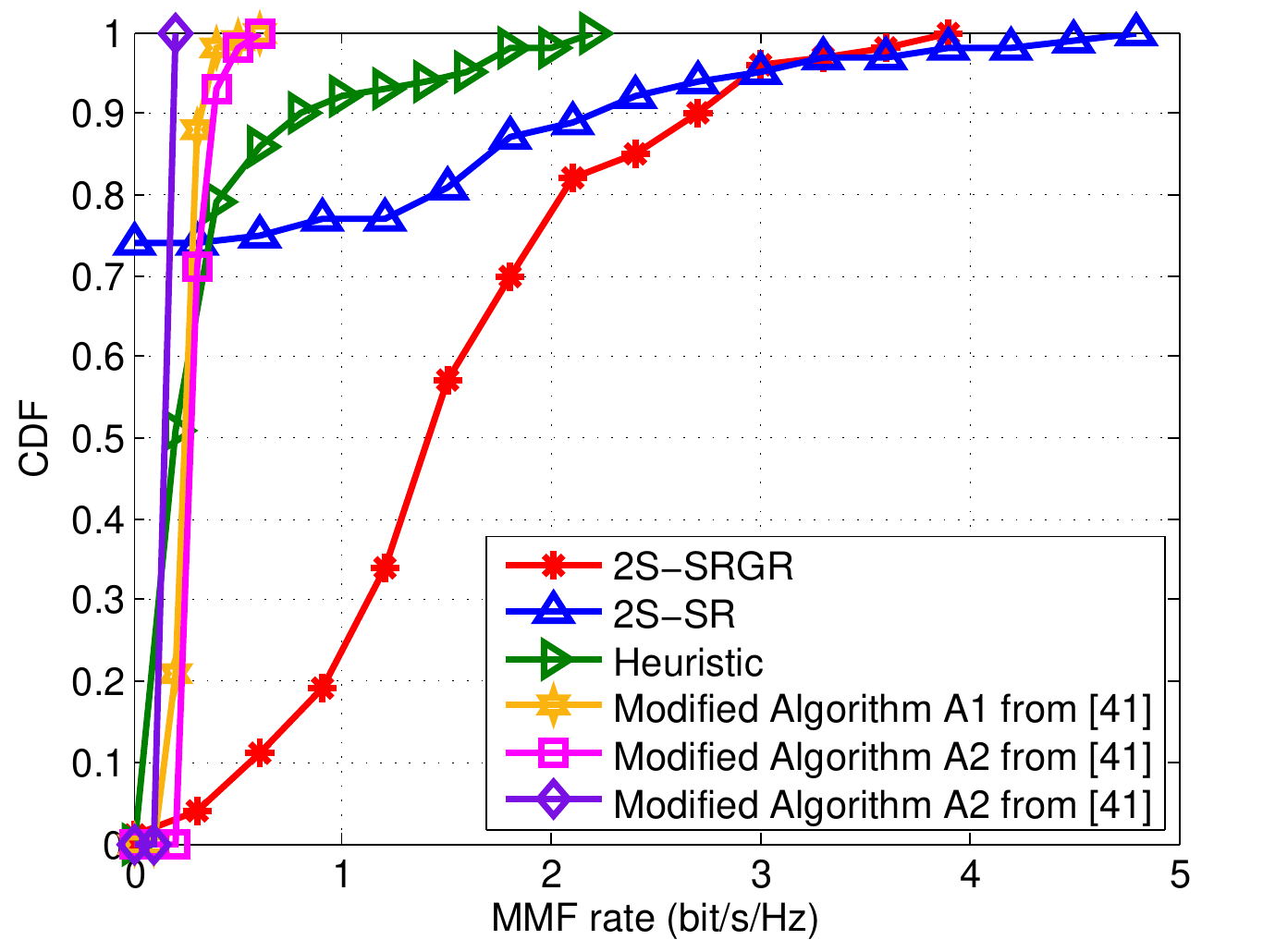}}\\
	\subfloat[$M=8,~B=64$]
	{\label{fig2b}\includegraphics[width=0.39\textwidth]{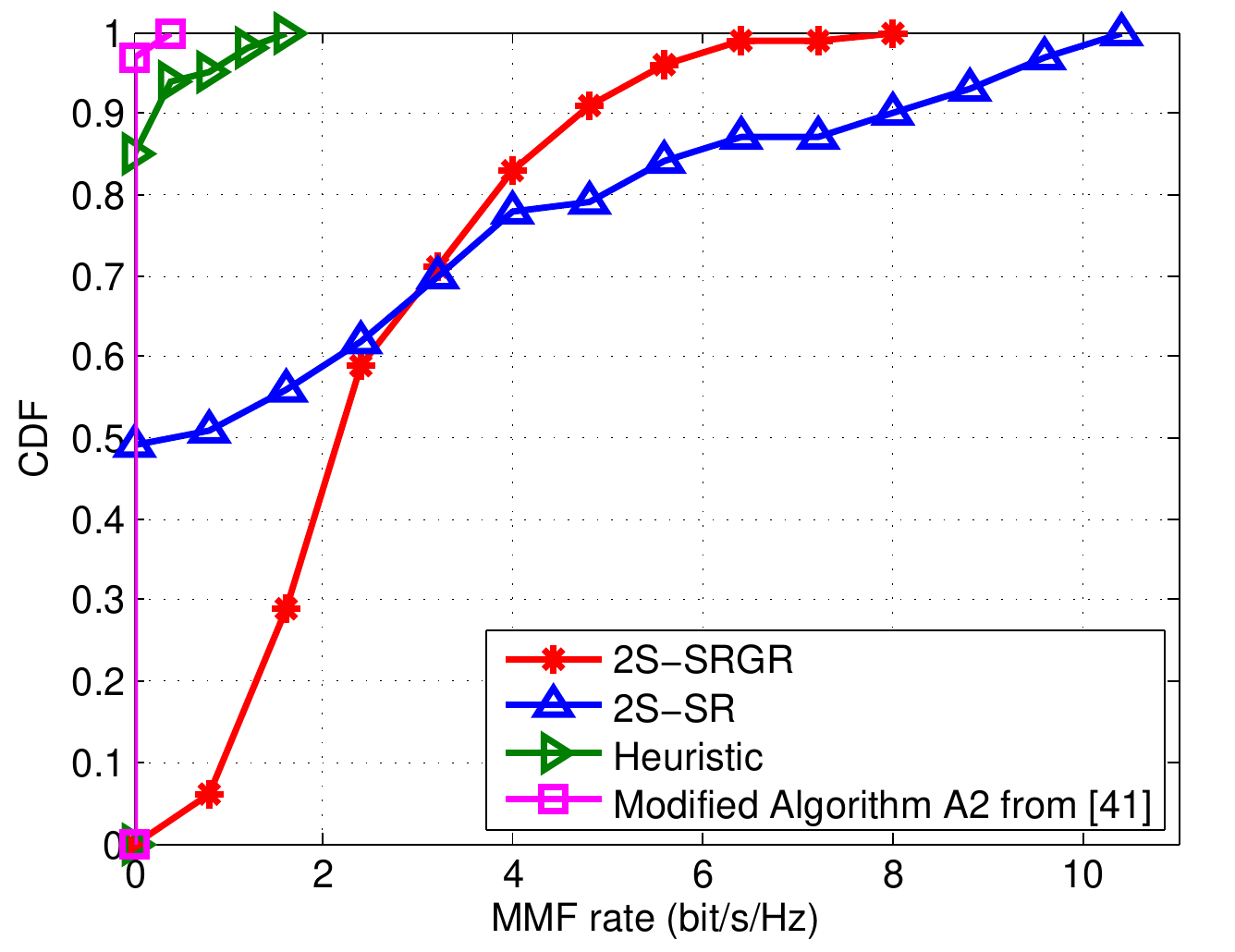}}
	\caption{CDFs of the MMF rates achieved by 2S-SR, 2S-SRGR, Heuristic, and the modified algorithms of A1-A3 from \cite{Alexand2016} in the non-full load case.}
	\label{fig:CDF_nfl}
\end{figure}
Fig. \ref{fig:CDF_nfl} shows the cumulative distribution functions (CDFs) of the MMF rates for $M=8$, $B=16$ (Fig. \ref{fig:CDF_nfl}(a)) and $M=8$, $B=64$ (Fig. \ref{fig:CDF_nfl}(b)). We do not show the results of SR since, as discussed in Section \ref{subsec:SR_SR}, the SR method almost cannot yield feasible solutions.
From both Fig. \ref{fig:CDF_nfl}(a) and Fig. \ref{fig:CDF_nfl}(b), one can observe that the 2S-SR method can still yield infeasible solutions, with probabilities higher than $74\%$ for the case of $B=16$ and $49\%$ for the case of $B=64$, respectively. In contrast, 2S-SRGR, Heuristic and the algorithms modified from A1-A3 in \cite{Alexand2016} can guarantee to provide feasible solutions with $100\%$ probability, while the 2S-SRGR method yields much higher MMF rates than Heuristic and the algorithms modified from A1-A3 in \cite{Alexand2016}. Specifically, for the case of $B=64$, 2S-SRGR achieves an MMF rate of 4.16 bit/s/Hz at the 80th percentile whereas Heuristic achieves an MMF rate of 0.17 bit/s/Hz only. The performance of the algorithms modified from A1-A3 in \cite{Alexand2016} is similar to Heuristic.
It can also be observed from Fig. \ref{fig:CDF_nfl}(a) and Fig. \ref{fig:CDF_nfl}(b) that the 2S-SR method may have a better chance to achieve a higher MMF rate than 2S-SRGR provided that the solution of 2S-SR is feasible. This is due to the reason that 2S-SRGR applies greedy rounding at each iteration.

\begin{figure}
    \linespread{1}
	\centering
    \subfloat[TDA and UP with uniform power allocation] {\includegraphics[width=0.4\textwidth]{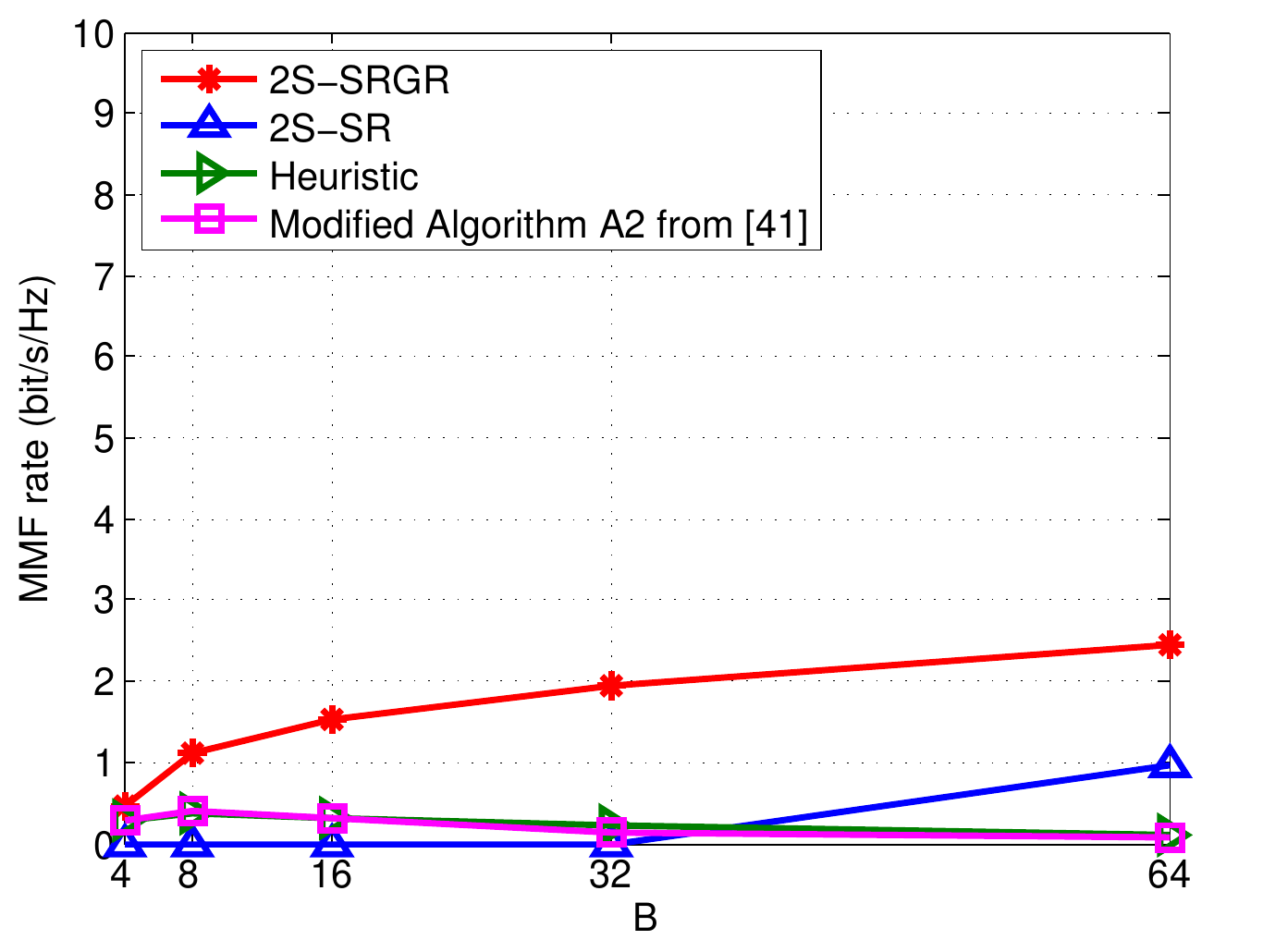}}\\
    \subfloat[TDA and UP with SCA power allocation]
	{\includegraphics[width=0.4\textwidth]{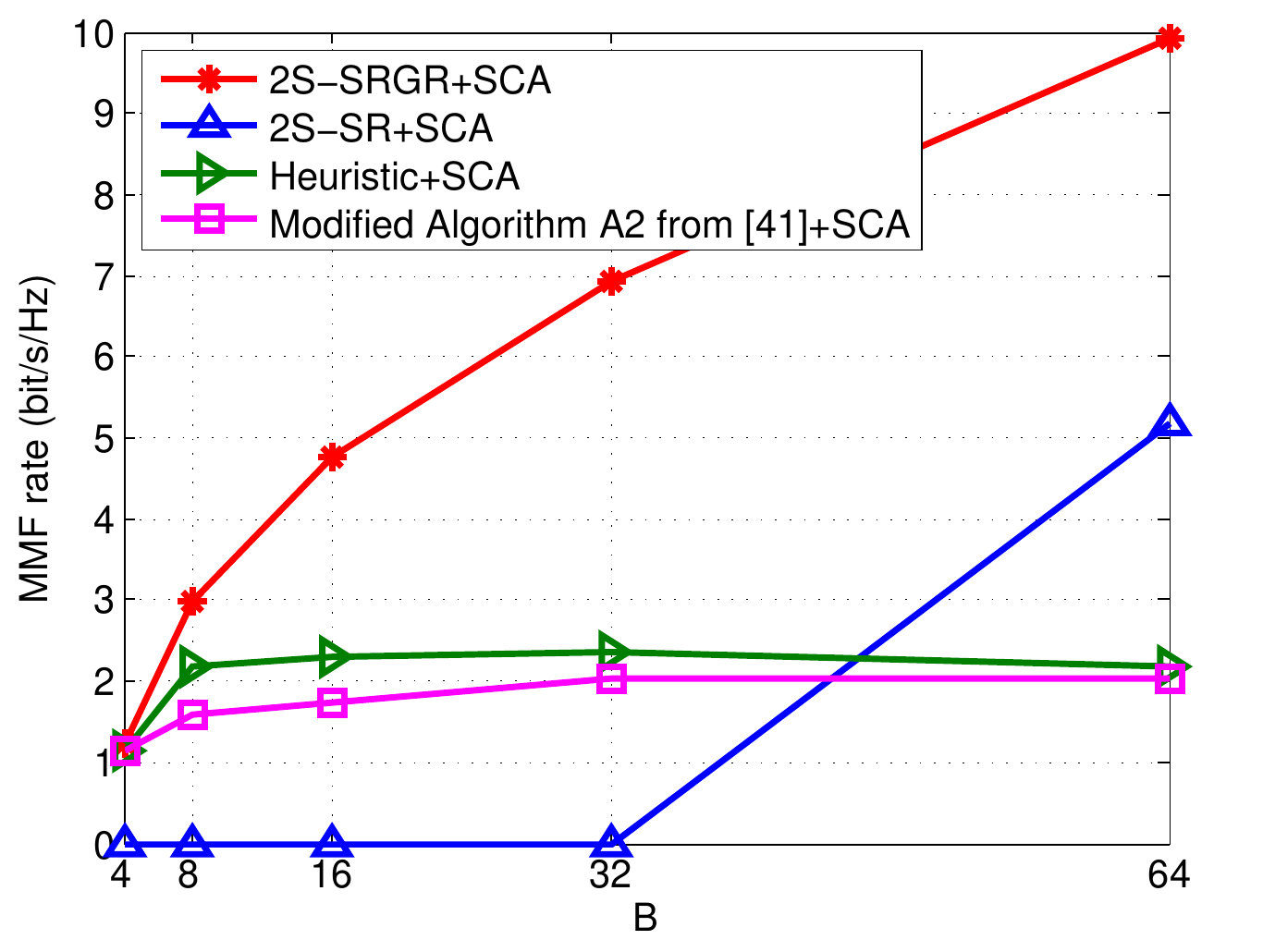}}\\
	\caption{The comparison of the MMF rates at the 50th percentile versus different RB number $B$: (a) 2S-SRGR, 2S-SR, Heuristic, and Modified Algorithm A2 from \cite{Alexand2016} are applied for TDA and UP with uniformed power allocation; (b) on the basis of (a), SCA is further applied for the power allocation.}
	\label{fig:mRate_nfl}
\end{figure}

Fig. \ref{fig:mRate_nfl} depicts the MMF rate versus the number of RBs at the 50th percentile, in which Fig. \ref{fig:mRate_nfl}(a) applies four methods for TDA and UP optimization with the uniformed power allocation and Fig. \ref{fig:mRate_nfl}(b) further allocates the power by SCA for the four methods. The number of UEs is 8 ($M=8$).
In Fig. \ref{fig:mRate_nfl}(a), the MMF rate of 2S-SRGR increases gradually as $B$ increases from $4$ to $64$; 2S-SR only gets non-zero 50th percentile MMF rate when $B=64$; the MMF rate returned by Heuristic and Modified Algorithm A2 from \cite{Alexand2016} are always slightly higher than zero. Note that the MMF rates of Heuristic and Modified Algorithm A2 from \cite{Alexand2016} goes up when $B\le 8$ and then goes down for $B>8$, which indicates that as $B$ increases, the benefit of allocating more RBs is gradually no longer evident given a fixed power budget at the BS.
Among the four methods, 2S-SRGR performs the best and much better than Heuristic and Modified Algorithm A2 from \cite{Alexand2016}. Specifically, the MMF rate of 2S-SRGR goes up from $0.45$ bit/s/Hz to $2.46$ bit/s/Hz as $B$ increases from  $4$ to $64$, which is much higher than that of Heuristic (always less than 0.39 bit/s/Hz) and Modified Algorithm A2 from \cite{Alexand2016} (always less than 0.41 bit/s/Hz). Also, the MMF rate of 2S-SRGR is almost 60\% higher than the one of 2S-SR when $M=8$, $B=64$.
In Fig. \ref{fig:mRate_nfl}(b), we omit the results obtained by applying the WMMSE method for the power allocation since as will be shown in Fig. \ref{fig:mRate_SIC}, the results by applying the WMMSE and the SCA methods are almost the same. Comparing with Fig. \ref{fig:mRate_nfl}(a), although the tendency of the four curves does not change, SCA has a fairly large MMF rate improvement, which shows the effectiveness of the power allocation.

\vspace{-1em}
\subsection{Case of $M=2B$}
Table \ref{tab:infea_FL} presents the infeasibility percentage of the six methods under the consideration in the full load scenario (i.e., $M=2B$). As seen, SR, 2S-SR and 2S-SRGR methods all exhibit high percentages to yield infeasible solutions, especially when $B$ is large (e.g, $B=8$).
In contrast, the 2S-IRMGR method based on $\ell_q$-norm regularization can effectively resolve the infeasibility issue. For example, the infeasibility percentage of 2S-IRMGR is always 0, whereas the ones of SR, 2S-SR and 2S-SRGR are at least 62\%, 30\% and 19\%, respectively, for $B=2$.
\begin{table*}
\linespread{1}
	\centering
    \caption{Infeasibility percentage of six algorithms in the full load case ($M=2B$).}
	\begin{tabular}{|c|c|c|c|c|c|c|c|}
		\hline
		Algorithms       &$B=2$ &$B=3$ &$B=4$ &$B=5$  &$B=6$  &$B=7$  &$B=8$\\
		\hline
		SR               &62\%&90\%&98\%&100\%&100\%&100\%&100\%\\
		2S-SR            &30\%&86\%&97\%&100\%&100\%&100\%&100\%\\
		2S-SRGR          &19\%&36\%&45\%&47\% &54\% &60\% &72\%\\
		2S-IRMGR/Heuristic/Modified
                         &\multirow{2}{*}{0\%}&\multirow{2}{*}{0\%}&\multirow{2}{*}{0\%}
                         &\multirow{2}{*}{0\%}&\multirow{2}{*}{0\%}&\multirow{2}{*}{0\%}
                         &\multirow{2}{*}{0\%}\\
         Algorithm from A2 in \cite{Alexand2016}
                         & & & & & & & \\
		\hline
	\end{tabular}
	\label{tab:infea_FL}
\end{table*}

\begin{figure}
	\centering
\linespread{1}
	\subfloat[$M=8,~B=4$]
	{\label{subfig2:cdf_M8B4}\includegraphics[width=.4\textwidth]{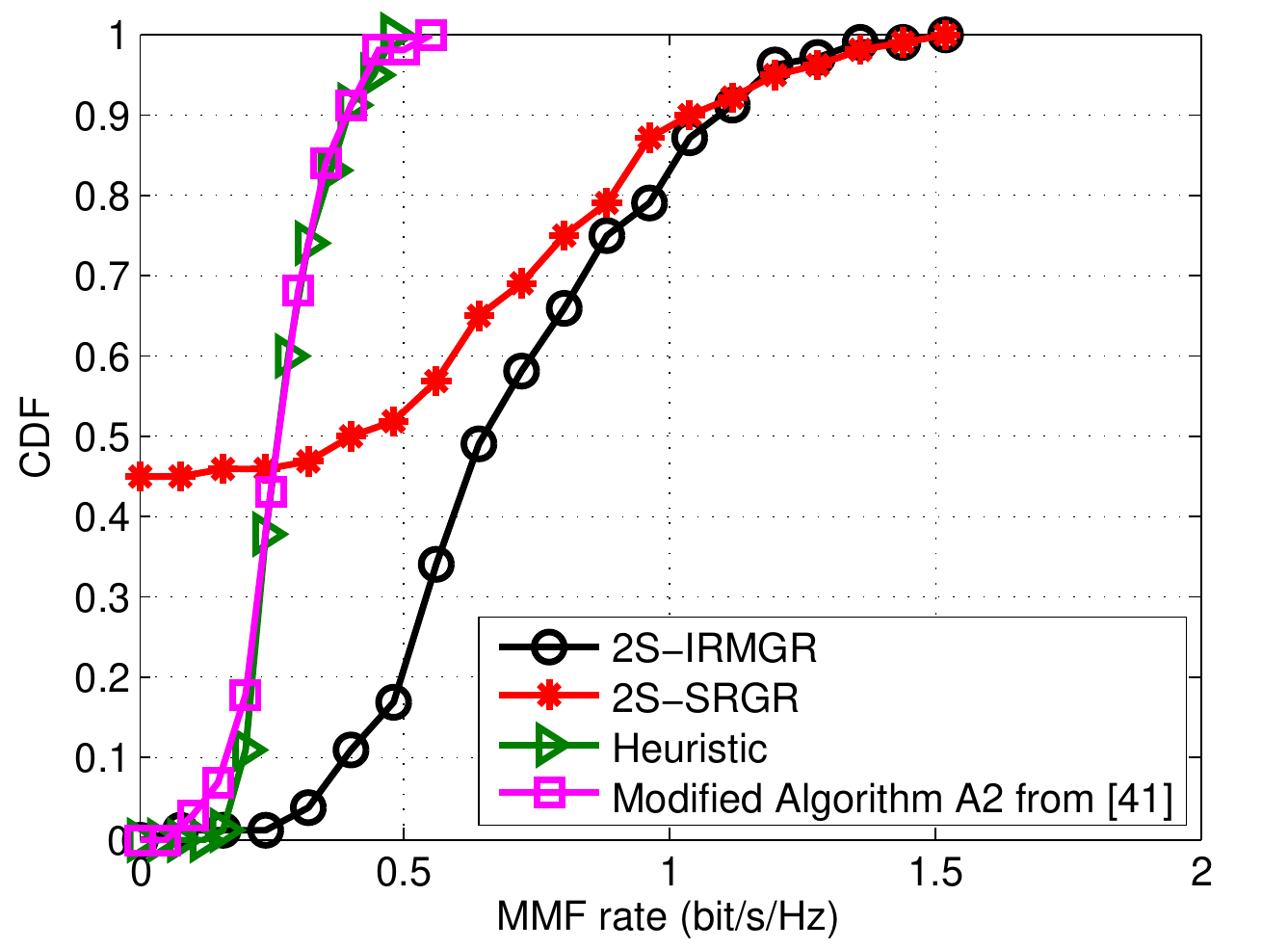}}\\
	\subfloat[$M=16,~B=8$]
	{\label{subfig2:cdf_M16B8}\includegraphics[width=.38\textwidth]{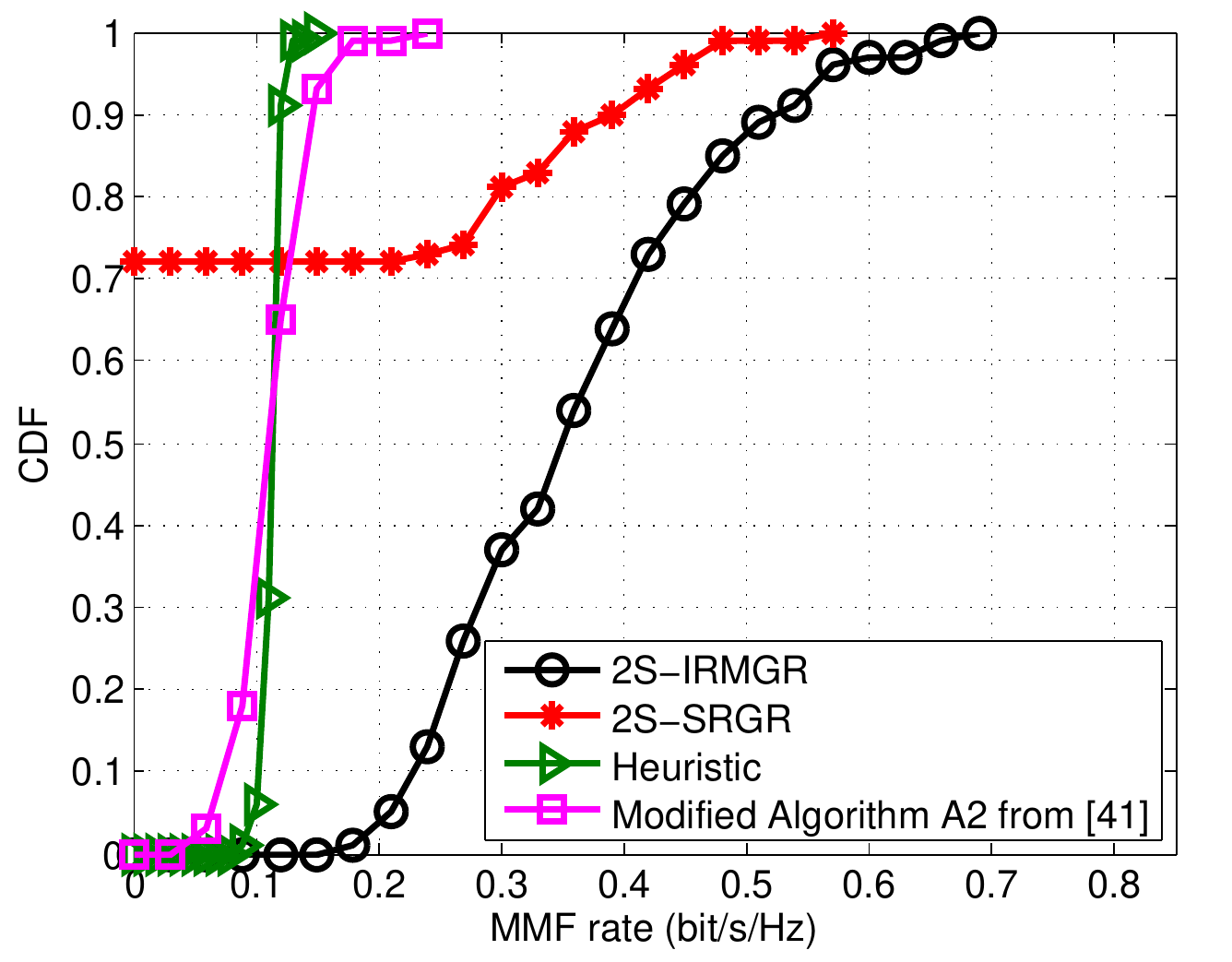}}
	\caption{CDFs of the MMF rates achieved by 2S-SRGR, 2S-IRMGR, Heuristic and Modified Algorithm A2 from \cite{Alexand2016} in the full-load case ($M=2B$).}
	\label{fig:CDF_fl}
\end{figure}

In Fig. \ref{fig:CDF_fl}, we display the CDFs of the average MMF rates achieved by 2S-SRGR, 2S-IRMGR, Heuristic and Modified Algorithm A2 from \cite{Alexand2016} in the full load case for $M=8$, $B=4$ (Fig. \ref{fig:CDF_fl}(a)) and $M=16$, $B=8$ (Fig. \ref{fig:CDF_fl}(b)). Again, we do not show the curves of SR and 2S-SR due to the high infeasibility percentage.
From Fig. \ref{fig:CDF_fl}(a) and Fig. \ref{fig:CDF_fl}(b), we can observe the consistent results that 2S-SRGR is likely to be infeasible with probabilities 45\% and 72\%, respectively.
In contrast, 2S-IRMGR, Heuristic and Modified Algorithm A2 from \cite{Alexand2016} can always guarantee a non-zero MMF rate. Comparing these three methods, the MMF rates of Heuristic and Modified Algorithm A2 from \cite{Alexand2016} are very close to each other, while 2S-IRMGR exhibits much higher MMF rate results. For example, in Fig. \ref{fig:CDF_fl}(a), among all channel realizations, the highest MMF rate achieved by Heuristic is 0.49 bit/s/Hz whereas more than 80\% MMF rates achieved by 2S-IRMGR are larger than 0.49 bit/s/Hz and the largest MMF rate achieved by 2S-IRMGR is 1.54 bit/s/Hz.
Also, in Fig. \ref{fig:CDF_fl}(b), the curve of 2S-IRMGR is always at the right bottom of the curve of 2S-SRGR, which means that the MMF rate of 2S-IRMGR is always higher than that of 2S-SRGR under the same CDF. This shows that 2S-IRMGR can guarantee every UE a non-zero MMF rate and does so without a high performance loss, which is unlike 2S-SRGR in the non-full load case.

\subsection{Computational Time}
In Table \ref{tab:complexity}, we present the computational time of Heuristic, 2S-SRGR, 2S-IRMGR and the off-the-shelf software CPLEX\cite{CPLEX}. The simulation is conducted on a two-core Intel i7-7500U CPU laptop. The first four rows show the results for the non-full load case with the number of the UEs $M=4$. Among the three algorithms, Heuristic consumes the least computational time but with the highest performance degradation (from 16.97 bit/s/Hz to 4.19 bit/s/Hz for $B=32$) compared to the optimal MMF rate obtained by CPLEX. 2S-SRGR is significantly more efficient than CPLEX for $B\ge16$ with about $30\%$ MMF rate loss. The last four rows displays the computational time for the full load case. As can be seen, a similar conclusion can be drawn.

\begin{table*}
\linespread{1}
	\centering
    \caption{Computational time and the MMF rate at the 50th percentile for the non-full and the full load case.}
	\begin{tabular}{|c||c|c|c||c|c|c|}
		\hline
		\multirow{1}{*}{}&\multicolumn{3}{c||}{Computational Time (s)} &\multicolumn{3}{c|}{MMF rate at the 50th percentile (bit/s/Hz)}\\
                         \cline{2-7}
        $M=4$          &Heuristic &2S-SRGR &CPLEX  &Heuristic &2S-SRGR &CPLEX\\
		\hline
        $B=4$       &$7.72\times10^{-3}$  &1.97   &0.34
                    &1.24                 &2.14   &3.02\\
		$B=8$       &$10.26\times10^{-3}$ &4.74   &0.95
                    &2.22                 &4.62   &6.53\\
        $B=16$      &$13.51\times10^{-3}$ &12.88  &19.95
                    &2.90                 &8.66   &11.21\\
        $B=32$      &$24.32\times10^{-3}$ &35.78  &$6.06\times10^3$
                    &4.19                 &10.99  &16.97\\
		\hline\hline
                         \cline{3-7}
       $M=2B$            &Heuristic &2S-IRMGR &CPLEX  &Heuristic &2S-IRMGR &CPLEX\\
		\hline
		$B=2$      &$7.42\times10^{-3}$  &1.04             &0.36
                   &0.55                 &1.21             &1.26\\
		$B=4$      &$10.28\times10^{-3}$ &19.28            &51.75
                   &0.28                 &0.70             &0.80\\
        $B=6$      &$13.59\times10^{-3}$ &90.66            &$1.62\times10^4$
                   &0.16                 &0.50             &0.64\\
        $B=8$      &$19.72\times10^{-3}$ &$1.36\times10^3$ &$3.72\times10^4$
                   &0.11                 &0.36             &0.49\\
        \hline
	\end{tabular}
	\label{tab:complexity}
\vspace{-1em}
\end{table*}

\subsection{Impact of Residual SI Channel Gain}
Fig. \ref{fig:mRate_SIC} depicts the statistical effect of the residual SI channel gain $\eta_b$ on the MMF rate in the case $M=8$, $B=64$. Three TDA and UP optimization algorithms (2S-SRGR, 2S-SR and Heuristic) followed by two power allocation algorithms (SCA and WMMSE) are compared. As can be seen from the figure, the MMF rates of all methods decrease as the residual SI channel gain increases.
Specifically, for the 2S-SRGR+WMMSE method, the MMF rate suffers a $99\%$ performance loss (from 8.90 bit/s/Hz to 0.02 bit/s/Hz) when $\eta_b$ changes from $-110$dB to $-60$dB.
Now, we focus on the comparison of different methods with the same residual SI channel gain. As expected, 2S-SRGR performs the best among the three TDA and UP optimization algorithms. With the same TDA and UP optimization algorithm, i.e., 2S-SRGR, the MMF rates of SCA and WMMSE are almost the same, which means that one can choose either one of them to perform power allocation.

\begin{figure}
	\centering
\linespread{1}
	\includegraphics[width=0.4\textwidth]{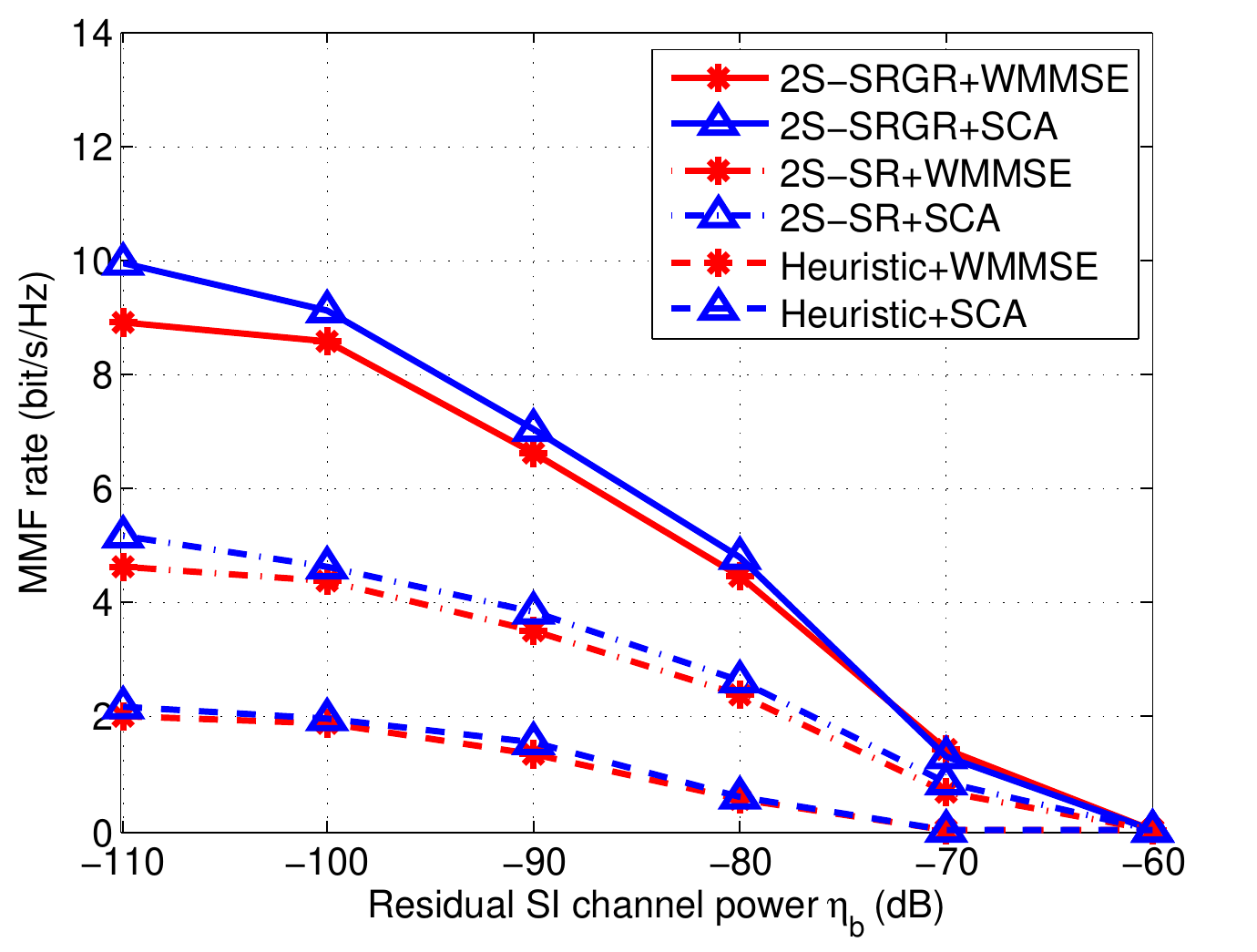}\\
	\caption{The comparison of the MMF rate at the 50th percentile in different residual SI channel gain. In this case, $M=8$, $B=64$, three TDA and UP methods: 2S-SRGR, 2S-SR, Heuristic and two power allocation methods: SCA and WMMSE are considered.}
	\label{fig:mRate_SIC}
\end{figure}
\section{Conclusion}\label{sec:con}
In this paper, we have investigated the joint TDA, UP and power allocation problem for maximizing the MMF rate in the FD multi-user OFDMA system.
The problem has been formulated as a  two-time-scale optimization problem where the TDA and UP variables are optimized to maximize a long-term MMF rate while the power allocation is optimized to maximize the short-term MMF rate. We have conducted the complexity analysis which shows that the considered joint MMF rate maximization problem is strongly NP-hard. The analysis suggests that the joint design problem is particularly difficult to solve in the full load case with $M=2B$.
To obtain the efficient TDA and UP solutions, we have developed several relaxation-and-rounding based algorithms. In particular, in order to avoid from the solutions that do not satisfy the HD transmission constraint of the UEs and to achieve a non-zero MMF rate, we have proposed a two-stage approach with an iterative rounding technique (Algorithm \ref{alg:GR}).
For the full load case, we have further proposed to tighten the relaxation by using the $\ell_q$-norm regularization and the IRM method  (Algorithm \ref{alg:2S-IRMGR}).
The simulation results have shown that the proposed algorithms are effective and greatly outperform the heuristic methods in the literature.

\ifCLASSOPTIONcaptionsoff
  \newpage
\fi

\bibliography{Ref5}
\begin{IEEEbiography}[{\includegraphics[width=1in,height=1.25in,clip,keepaspectratio]{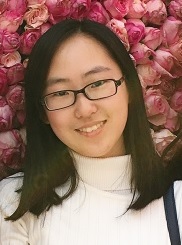}}]{Xiaozhou Zhang}
received the B.S. degree in communication engineering from Beijing Jiaotong University (BJTU), Beijing, China, in 2014. From 2015 to 2017, she visited The Chinese University of Hong Kong, Shenzhen, China. In 2018, she visited Singapore University of Technology and Design. She is currently pursuing the Ph.D. degree at the State Key Laboratory of Rail Traffic Control and Safety, Beijing Jiaotong University. Her research interests include physical layer resource allocation in full-duplex systems, optimization problems in NOMA, and the analysis of wireless communication systems.
\end{IEEEbiography}

\begin{IEEEbiography}[{\includegraphics[width=1in,height=1.25in,clip,keepaspectratio]{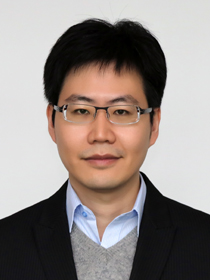}}]{Tsung-Hui Chang}
(S'07--M'08) received the B.S. degree in electrical engineering and the Ph.D. degree in communications engineering from the National Tsing Hua University (NTHU), Hsinchu, Taiwan, in 2003 and 2008, respectively. From 2012 to 2015, he was an Assistant Professor with the Department of Electronic and Computer Engineering, National Taiwan University of Science and Technology (NTUST), Taipei, Taiwan. In August 2015, Dr. Chang joined the School of Science and Engineering, The Chinese University of Hong Kong, Shenzhen, China, as an Assistant Professor, and since August 2018, was promoted to be an Associate Professor. Prior to being a faculty member, Dr. Chang was a visiting PhD student of the University of Minnesota, Minneapolis, MN, USA (2016.09-2018.02), a postdoctoral researcher with NTHU (2008-2011) and with the University of California, Davis, CA, USA (2011-2012). His research interests include signal processing and optimization problems in data communications, machine learning and big data analysis.

Dr. Chang received the Young Scholar Research Award of NTUST in 2014, IEEE ComSoc Asian-Pacific Outstanding Young Researcher Award in 2015, and the IEEE Signal Processing Society Best Paper Award in 2018. He served as an Associate Editor of IEEE TRANSACTIONS ON SIGNAL PROCESSING and IEEE TRANSACTIONS ON SIGNAL AND INFORMATION PROCESSING OVER NETWORKS (2015-2018).
\end{IEEEbiography}

\begin{IEEEbiography}[{\includegraphics[width=1in,height=1.25in,clip,keepaspectratio]{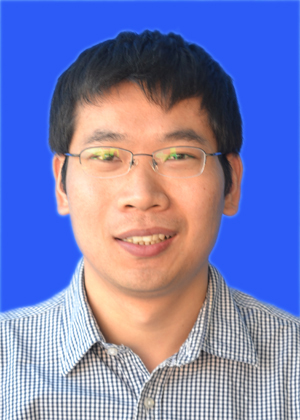}}]{Ya-Feng Liu}
(M'12--SM'18) received the B.Sc. degree in applied mathematics in 2007 from Xidian University, Xi'an, China, and the Ph.D. degree in computational mathematics in 2012 from the Chinese Academy of Sciences (CAS), Beijing, China. During his Ph.D. study, he was supported by the Academy of Mathematics and Systems Science (AMSS), CAS, to visit Professor Zhi-Quan (Tom) Luo at the University of Minnesota (Twins Cities) from February 2011 to February 2012. After his graduation, he joined the Institute of Computational Mathematics and Scientific/Engineering Computing, AMSS, CAS, Beijing, China, in July 2012, where he is currently an Associate Professor. His main research interests are nonlinear optimization and its applications to signal processing, wireless communications, and machine learning. He is especially interested in designing efficient algorithms for solving optimization problems arising from the above applications.

Dr. Liu has served as a guest editor of the Journal of Global Optimization since 2016. He is a recipient of the Best Paper Award from the IEEE International Conference on Communications (ICC) in 2011 and the Best Student Paper Award from the International Symposium on Modeling and Optimization in Mobile, Ad Hoc and Wireless Networks (WiOpt) in 2015. He also received the Chen Jingrun Star Award from the AMSS and the Science and Technology Award for Young Scholars from the Operations Research Society of China in 2018.
\end{IEEEbiography}

\begin{IEEEbiography}[{\includegraphics[width=1in,height=1.25in,clip,keepaspectratio]{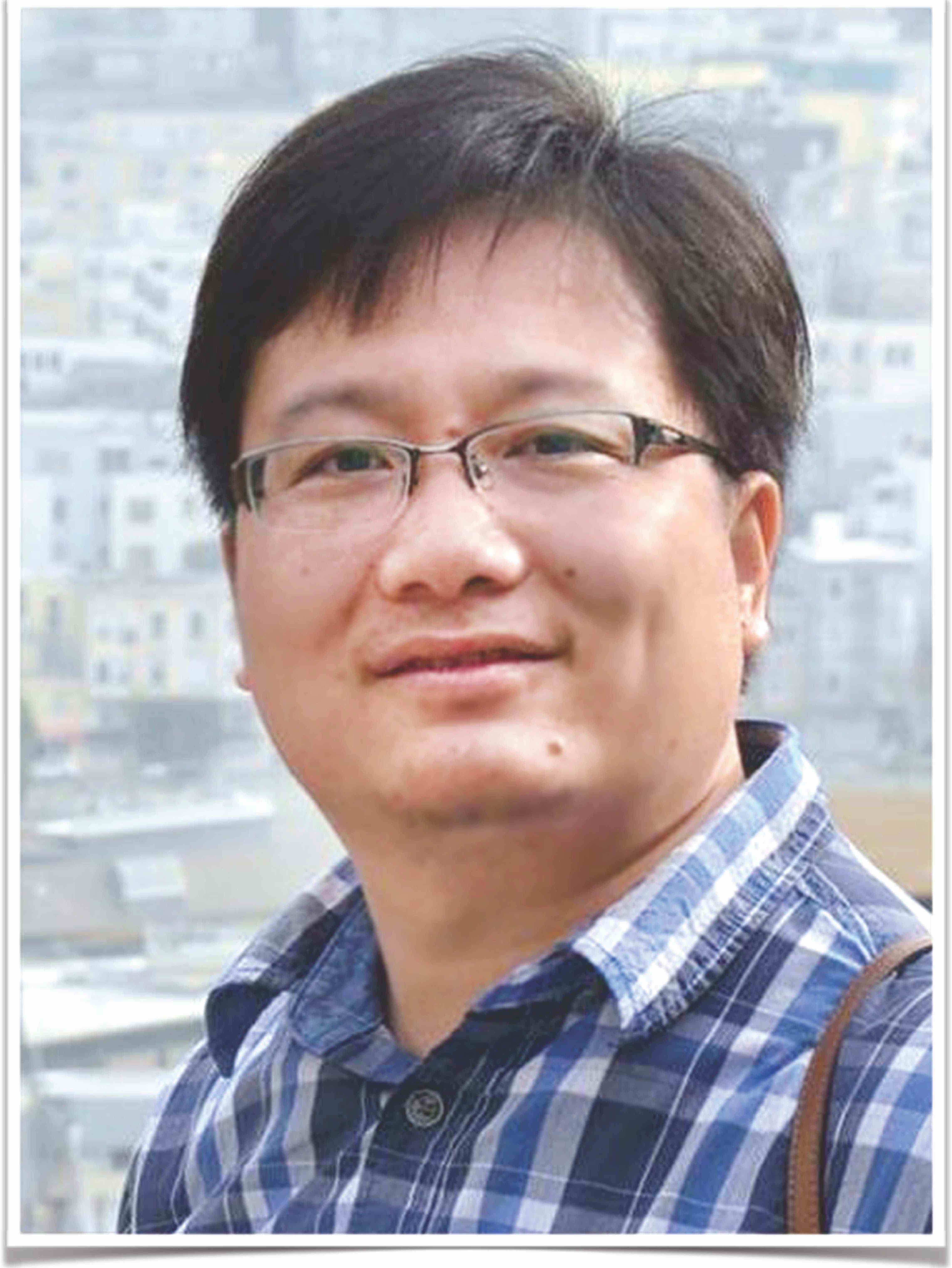}}]{Chao Shen}
(S'12--M'13) received the B.S. degree in communication engineering and the Ph.D. degree in signal and information processing from Beijing Jiao- tong University (BJTU), Beijing, China, in 2003 and 2012, respectively. He held a post-doctoral position at BJTU. He was a Visiting Scholar with the University of Maryland at College Park, College Park, MD, USA, from 2014 to 2015, and The Chinese University of Hong Kong, Shenzhen, from 2017 to 2018. Since 2015, he has been an Associate Professor with the State Key Laboratory of Rail Traffic Control and Safety, BJTU.
His current research interests focus on the ultra-reliable and low-latency communications, unmanned aerial vehicle-enabled wireless communications, and energy-efficient wireless communications for 5G/B5G communications.
\end{IEEEbiography}

\begin{IEEEbiography}[{\includegraphics[width=1in,height=1.25in,clip,keepaspectratio]{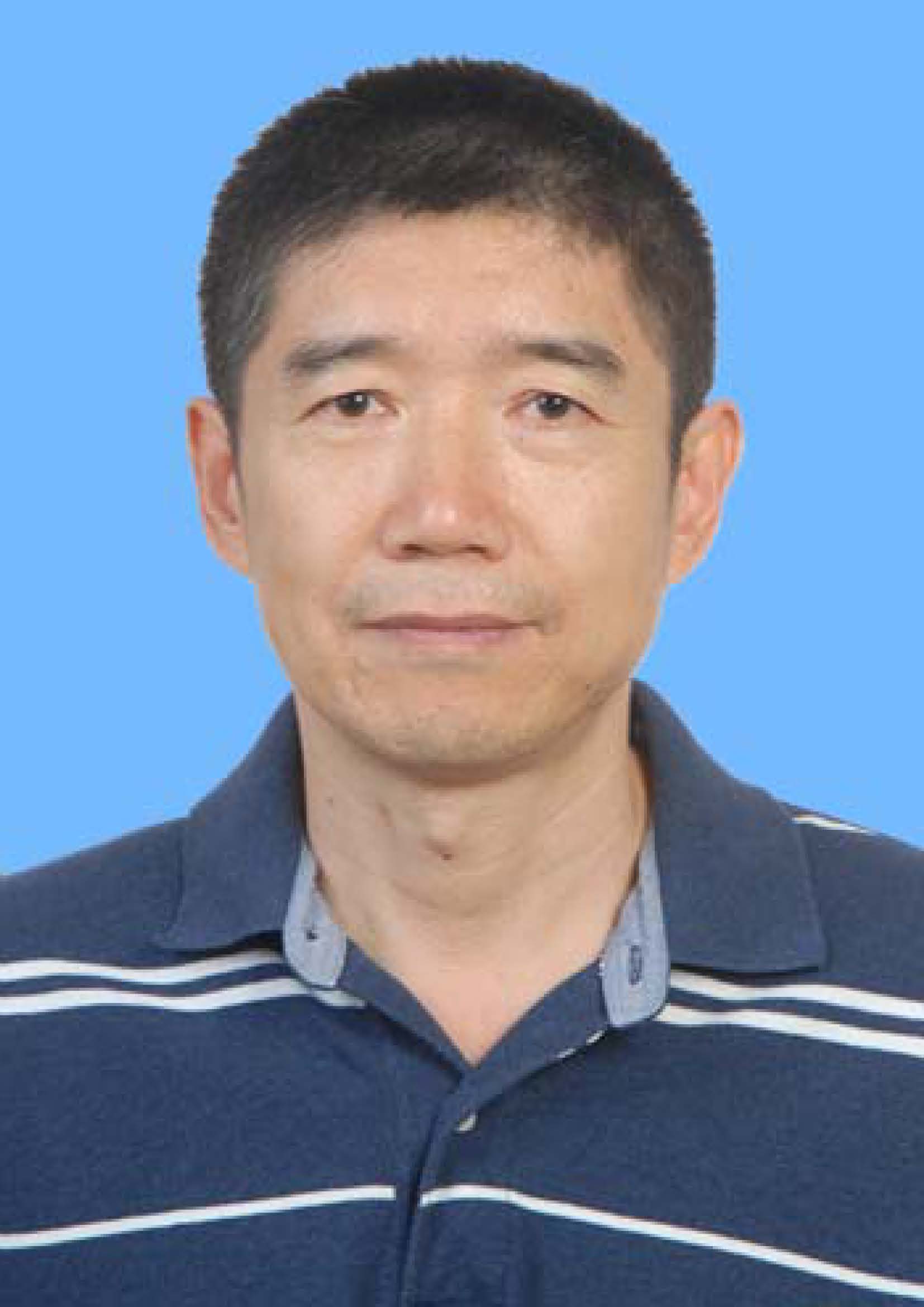}}]{Gang Zhu}
received the M.S. degree in 1993 and PhD degree in 1996 from Xi'an Jiaotong University, Xi'an, China. Since 1996, he has been with Beijing Jiaotong University, Beijing, China, where he is currently a professor. During 2000 to 2001, he was a Visiting Scholar in Department of Electrical and Computer Engineering, University of Waterloo, Canada.
His current research interests include resource management in wireless communications, short distance wireless communications, and global system for mobile communications for railways (GSM-R). He received Top Ten Sciences and Technology Progress of Universities in China in 2007 and First Class Award of Science and Technology in Railway in 2009.
\end{IEEEbiography}

\end{document}